\documentclass[10pt, journal]{IEEEtran}
\usepackage[numbers]{natbib}
\makeatletter
\renewcommand{\@biblabel}[1]{\makebox[5em][l]{[#1]}}
\makeatother
\usepackage{pgfplots}
\usepackage{tikz}
\usetikzlibrary{calc}
\usepgfplotslibrary{groupplots}
\pgfplotsset{compat=1.18}
\usepackage{titlesec}

\titleformat{\paragraph}[runin]
  {\normalfont\normalsize\bfseries} 
  {}                               
  {0pt}                            
  {}                               

\titlespacing*{\paragraph}
  {0pt}                            
  {1.5ex plus 0.5ex minus .2ex}    
  {1em}                            

\usepackage{amsmath,amsfonts,bm}

\def\eqref#1{equation~\ref{#1}}

\def\1{\bm{1}}

\DeclareMathAlphabet{\mathsfit}{\encodingdefault}{\sfdefault}{m}{sl}
\SetMathAlphabet{\mathsfit}{bold}{\encodingdefault}{\sfdefault}{bx}{n}

\newcommand{\E}{\mathbb{E}}

\usepackage{array}
\usepackage{url}
\usepackage[utf8]{inputenc}
\usepackage[T1]{fontenc}
\usepackage{bbm}
\usepackage{ifthen}
\usepackage{mathtools}              
\usepackage{amssymb}
\usepackage{amsthm}
\usepackage{tabulary}
\usepackage{subcaption}
\usepackage{graphicx}

\usepackage{enumitem}
\usepackage{tikz}
\usetikzlibrary{positioning, fit} 
\usepackage[ruled]{algorithm2e}
\usepackage{color}
\usepackage{xcolor}
\usetikzlibrary{arrows.meta, positioning}

\usepackage[
    colorlinks=true,
    linkcolor=green,
    citecolor=blue,
    urlcolor=blue,
    filecolor=green,
    breaklinks=true
]{hyperref}
\usepackage[nameinlink,capitalize,noabbrev]{cleveref}
\makeatletter
\renewcommand{\eqref}[1]{\textup{\tagform@{\ref{#1}}}}
\makeatother
\hypersetup{
pdftitle={Local and Global Contraction Principles for MCMC Mixing},
pdfsubject={,},
pdfauthor={Alireza Daeijavad and Shahab Asoodeh},
pdfkeywords={MCMC mixing, contraction coefficients, hockey-stick divergence, projected Langevin Monte Carlo, independent Metropolis--Hastings},
}

\newcommand{\sE}{\mathsf{E}}

\newcommand{\sK}{\mathsf{K}}
\newcommand{\renyi}{R\'enyi}
\newcommand{\egamma}{$\mathsf{E}_\gamma$}
\newcommand{\TV}{\mathsf{TV}}
\newcommand{\msmooth}{$M$-smooth}

\newtheorem{definition}{Definition}
\newtheorem{theorem}{Theorem}
\newtheorem{example}{Example}
\newtheorem{remark}{Remark}
\newtheorem{proposition}{Proposition}
\newtheorem{corollary}{Corollary}
\newtheorem{lemma}{Lemma}

\usepackage{xcolor}

\title{Local and Global Contraction Principles for MCMC Mixing}

\usepackage{lipsum}

\makeatletter
\def\blfootnote{\gdef\@thefnmark{}\@footnotetext}
\makeatother

\author{%
Alireza Daeijavad and
Shahab Asoodeh \\

\thanks{${}^\dagger$ A. Daeijavad and S. Asoodeh are with the Department of Computing and Software, McMaster University, Hamilton,
ON L8S 4K1, Canada (email: {daeijava, asoodehs}@mcmaster.ca).}
}

\date{}

\begin{document}

\onecolumn

\maketitle

\begin{abstract}
We develop a contraction-based framework for proving mixing-time bounds for Markov chain Monte Carlo algorithms. The framework is built around global and local contraction coefficients of Markov kernels under the $\mathsf E_\gamma$-divergence with $\gamma\ge1$. For projected Langevin Monte Carlo on a compact convex domain, we show that Gaussian smoothing yields an explicit global contraction coefficient for the $\mathsf E_\gamma$-divergence. This gives a direct proof of exponential convergence to the discretized stationary distribution for general smooth, possibly non-convex potentials. The rate is explicit, accommodates arbitrary random-batch sampling schemes, and yields convergence guarantees for several divergences, including KL, $\chi^2$, and R\'enyi divergences. 
For independent Metropolis--Hastings with target $\pi$, proposal $q$, and unbounded importance weight $w=d\pi/dq$, global contraction coefficients are typically trivial. We therefore introduce a local contraction coefficient on the core $C_R=\{w\le R\}$ and prove that it controls the rejection profile on the core. This yields warm-start convergence bounds governed by the local contraction coefficient and the tail profile $H_R=\pi(w>R)$, recovering sharp existing moment-based convergence rates when $\mathbb E_q[w^p]<\infty$ for some $p>1$, while remaining effective in heavy-tailed regimes where no finite moment of order $p>1$ exists.
\end{abstract}

\section{Introduction}

Sampling from a target distribution $\pi$ is a central problem in statistics, machine learning, and scientific computing. Markov chain Monte Carlo (MCMC) methods approach this task by constructing a Markov chain whose distribution converges to $\pi$ or, for discretized algorithms, to a controlled stationary approximation of $\pi$. A fundamental question is therefore quantitative: how quickly does the law of the chain approach stationarity, and in which metric? This question is especially delicate for modern sampling algorithms, whose transition kernels often combine deterministic maps, Gaussian noise, projections, stochastic gradients, or Metropolis--Hastings accept-reject steps. 

This paper develops a strong data-processing inequality (SDPI) perspective on this problem. Given a Markov kernel $\sK$ and a divergence $D$, an SDPI bound has the form $D(\mu\sK\|\nu\sK)\le \eta D(\mu\|\nu)$ for some $\eta<1$; the smallest such $\eta$ is the global contraction coefficient of $\sK$.  Such an inequality converts a one-step contraction property into a mixing bound by iteration. The main advantage of this viewpoint is structural: rather than proving convergence separately for each algorithm, one identifies where contraction enters the transition kernel and uses data processing to propagate it through the rest of the update.

The divergence used throughout is the $\mathsf E_\gamma$-divergence (also known as hockey-stick divergence) for $\gamma\ge1$. This family is particularly well suited to SDPI analysis. It includes total variation as the endpoint, namely $\mathsf E_1(\mu\|\nu)=\mathsf{TV}(\mu, \nu)$, and more  importantly, its full profile controls many familiar divergences through integral representations (see identity in \eqref{f_representation}). Thus, an $\mathsf E_\gamma$ contraction theorem provides a mechanism for proving convergence in a broad class of divergences relevant to sampling, including KL and $\chi^2$-divergences.

We apply this perspective to two settings that require different forms of contraction. The first is projected Langevin Monte Carlo (P-LMC) on a compact convex set. In this case, the update decomposes into a drift step, Gaussian smoothing, and a projection. The projection and drift steps are non-expansive or controlled by the data processing inequality, while the Gaussian smoothing step provides a strict global SDPI on compact sets. This leads to a global contraction argument. The second setting is independent Metropolis--Hastings (IMH). Here a global SDPI is typically unavailable: when the importance weight $w=d\pi/dq$ is unbounded, where $\pi$ is the target and $q$ is the proposal, the chain can reject with probability arbitrarily close to one, and any uniform one-step contraction becomes trivial. For this reason, the appropriate analogue is \textit{local} SDPI: we prove contraction on the high-probability core $C_R=\{w\le R\}$ and control the remaining error through the tail profile $H_R=\pi(w>R)$.

The two analyses share the same contraction principle, but use it in different regimes. For P-LMC, compactness and Gaussian smoothing yield a global contraction coefficient. For IMH, the same SDPI principle must be localized: the local coefficient controls the holding probability on the core $C_R$, while $H_R$ measures the price of leaving that core. Thus the paper develops two complementary uses of SDPI for mixing: global contraction when the kernel has global smoothing, and local contraction when the kernel mixes only on a high-probability region. Detailed definitions of these algorithms are provided in Sections~\ref{Sec:PLMC} and \ref{sec:MH}.

\paragraph{Contributions.}
Our first contribution is a global SDPI analysis of P-LMC. We show that for smooth potentials on a compact convex set, the P-LMC kernel contracts $\mathsf E_\gamma$-divergence exponentially fast for every $\gamma\ge1$. Notably, this global contractivity does not require convexity of the potential. The contraction coefficient is explicit and depends on the diameter of the drifted set before the Gaussian smoothing step. This gives a clean explanation of why P-LMC admits a global contraction analysis on compact domains: Gaussian noise contracts the hockey-stick profile once the pre-noise image has bounded diameter.

Our second contribution is an average-case convergence bound for the stochastic-gradient P-LMC. More precisely, we prove that both $\mathsf E_\gamma(\mu_n\|\pi^\eta)$ and $\mathsf E_\gamma(\pi^\eta\|\mu_n)$ converge to zero exponentially fast for every $\gamma\geq 1$, where  $\mu_n$ is the distribution of the $n$th step of the chain and $\pi^\eta$ is the stationary distribution of the discretized chain with $2\eta$ as the Gaussian noise parameter. When the drift is computed using a random batch, the contraction coefficient can be averaged over the batch distribution. This separates the effect of the sampling scheme from the worst-case smoothness bound and yields sharper guarantees when different batches have different smoothness constants. A worst-case corollary recovers a simpler bound independent of the batching rule.

Our third contribution is a framework for transferring from the hockey-stick convergence to broader divergences. Since the P-LMC result controls both profiles $\mathsf E_\gamma(\mu_n\|\pi^\eta)$ and $\mathsf E_\gamma(\pi^\eta\|\mu_n)$, the integral representation of general $f$-divergences in terms of $\sE_\gamma$-divergence yields convergence bounds for a broad class of $f$-divergences with twice differentiable $f$, including KL-divergence, $\chi^2$-divergence, and R\'enyi divergence. This highlights the advantage of proving a full $\mathsf E_\gamma$-divergence profile bound rather than a single TV convergence result.

Our fourth contribution is a local-SDPI analysis of IMH. For the core $C_R=\{w\le R\}$, we define a local hockey-stick contraction coefficient $\rho_\alpha(R)$ for the truncated target $\pi_R=\pi(\cdot\mid C_R)$. Under a non-atomic proposal, this coefficient directly controls the pointwise rejection probability $r(x)$ on the core, namely,  $r(x)\le \rho_\alpha(R)$ for $x\in C_R$. We further prove the explicit bound $\rho_\alpha(R)\le 1-h_R/R$, where $h_R=\pi(C_R)$. These two results, together with a sharp rejection-profile estimate, yield a parametric convergence under an $L$-warm start:
$$\mathsf E_\gamma(\mu_n\|\pi)\lesssim (L+1)\Big[e^{-\frac{nh_R}{R}}+H_R\Big],$$
for every $R$ and $\gamma\geq 1$. 
This formulation recovers the moment-based rate of \cite{deligiannidis2024importance} as a special case under warm starts.  Moreover, we show that such hockey-stick convergence bound  can be directly translated into KL and $\chi^2$-divergences which, unlike the framework developed for P-LMC, follows from the assumption of warm starts. 

\paragraph{Relation to prior work on P-LMC.}
The closest line of work for the P-LMC part is the sharp mixing-time analysis of \cite{altschuler2022resolving}. Their work resolves the mixing time of projected Langevin algorithms in the convex and smooth setting and introduces tools from differential privacy into sampling.  
While powerful, this machinery is restricted to convex potentials and does not seem to extend to the non-convex setting. Additionally, their proof technique can only account for a particular sampling scheme, namely, sampling without replacement (i.e., fixed-size mini batch). In contrast, our convergence results hold for non-convex potentials and allow  arbitrary sampling schemes (e.g., Poisson sampling, which is widely used in optimization and sampling literature.) 

Compared to other existing results (see Table~\ref{tab_overview}), our contributions offer three key advantages: (1) the derived bounds apply to a broader class of potentials, requiring only smoothness, whether the potentials are convex or non-convex, (2) the results hold for a wide range of $f$-divergences, including KL divergence, R\'enyi divergence, TV distance, and Hellinger distance, and (3) our proof technique is independent of the batching scheme used to construct the batch at each iteration. As a result, various sampling strategies, such as Poisson sampling and sampling without replacement, can be applied for selecting the batches.

\paragraph{Relation to prior work on non-convex Langevin.}
There is a large literature on non-convex sampling for Langevin dynamics and its discretizations. Existing analyses often rely on functional inequalities or dissipativity-type assumptions, such as log-Sobolev, Poincar\'e, weak Poincar\'e, or related inequalities, and obtain convergence in Wasserstein, KL, $\chi^2$, R\'enyi, or general $f$-divergences \cite{raginsky2017non,vempala2019advances,erdogdu2021convergence,erdogdu2022chisquared,chewi2022analysis,Hosseini2023Towards,mitra2025fast}. These works primarily study unprojected Langevin dynamics or LMC. The closest projected non-convex result is \cite{lamperski2021projected}, who analyze P-LMC in $W_1$ under mild non-convex assumptions through a comparison between continuous and discrete processes. Our analysis is different: it works directly with the discrete projected chain, targets the biased stationary distribution $\pi^\eta$, and obtains divergence-profile bounds from a one-step SDPI argument. See Table~\ref{tab_overview} for clearer comparison and also Appendix~\ref{whole_literature} for more comprehensive literature review. 

\setlength{\tabcolsep}{2pt}

\begin{table}[t]
    \caption{Summary of convergence results for Langevin dynamics and related algorithms, with 'Type' indicating convergence to the target or biased distribution (i.e., stationary distribution of the discretized variant). }
    \label{tab_overview}
    \centering
    \begin{tabulary}{\columnwidth}{|>{\centering\arraybackslash}m{25mm}|>{\centering\arraybackslash}m{12mm}|>{\centering\arraybackslash}m{12mm}|>
    {\centering\arraybackslash}m{35mm}|>{\centering\arraybackslash}m{20mm}|C|}
        \hline
        Reference & Algo. & Convex & Other Assumptions & Metric & Type\\
        \hline
        \cite{Hosseini2023Towards} & LD & No & WPI, s-H\"older & \renyi{} & to target\\
        \hline
        \cite{raginsky2017non} & LMC & No & LSI, \msmooth{}, dissipative & $W_2$ & to target\\
        \hline
        \cite{Hosseini2023Towards} & LMC & No & WPI, s-H\"older & \renyi{} & to target\\
        \hline
        \cite{mitra2025fast} & LMC & No & \msmooth{}, $f$-Sobolev Inequality & $f$-divergence & to biased\\
        \hline
        \cite{lamperski2021projected} & P-LMC & No & \msmooth{}, uniform sub-Gaussian gradients & $W_1$ & to target\\
        \hline
        \cite{bubeck2018sampling} & P-LMC & Yes & \msmooth{}, Lipschitz & TV & to target\\
        \hline
        \cite{altschuler2022resolving} & P-LMC & Yes & \msmooth{} & TV & to biased\\
        \hline
        Ours & P-LMC & No & \msmooth{} & $f$-divergence & to biased\\
        \hline
    \end{tabulary}
\end{table}

\paragraph{Relation to MH and drift-minorization.}
The Metropolis--Hastings literature has long emphasized that convergence depends on the compatibility between the proposal and the target. 
If $\pi\ll q$, then the IMH chain is $\pi$-irreducible, aperiodic, and $\pi$-invariant and thus $\TV(\sK^n(x, \cdot), \pi)\to 0$ $\pi$-a.s.\ as $n\to \infty$, where $\sK^n$ denotes the $n$-step transition kernel.  More precisely, \cite{MengersonTweedie} showed that uniform ergodicity is essentially equivalent to the proposal dominating the target uniformly, or equivalently, to the importance weight $w=d\pi/dq$ being bounded. This identifies the globally contractive regime: if $w$ is bounded, then taking $R=\|w\|_\infty$ makes the core $C_R=\{w\le R\}$ equal to the whole state space and gives $H_R=0$. 
However, when it comes to non-asymptotic behavior, there is an important distinction between two cases: either the weight is bounded, in which case the chain is geometrically ergodic with exact rates obtained in \cite{wang2022exact,brown2024exact}, or the weight is unbounded and the convergence cannot be geometric \cite{roberts2011quantitative,Andrieu_IMH,deligiannidis2024importance}. 

In fact, when $w$ is unbounded, uniform ergodicity fails, and convergence is governed by the tail behavior of $w$ and by repeated rejections. This was recently formalized by \cite{deligiannidis2024importance} who provided polynomial bounds on the total variation distance to stationarity under moment constraints on $w$: $\E_q[w^p]<\infty$ for some $p>1$. 
Their proof uses a common-randomness coupling in which two IMH chains share the same proposals and acceptance variables. Under this construction, convergence is governed by how long the chain started from the larger importance weight keeps rejecting, making the rejection profile $r(x)^n$ and its stationary average $\int r^nd\pi$ the central finite-time quantities. Under finite moment assumptions on $w$, they obtain a sharp polynomial TV bound. We recover the same moment-based rate under warm starts, but our formulation is stated directly in terms of the tail profile $H_R=\pi(w>R)$ rather than the moments. This makes the bound tail-adaptive: moment assumptions are only one way to control $H_R$, and sharper model-specific tail estimates can be inserted directly. In particular, Example~\ref{ex:no-finite-moment-imh} gives a case where $\mathbb E_q[w^p]=\infty$ for every $p>1$, so finite-moment polynomial bounds of \cite{deligiannidis2024importance} do not provide a quantitative conclusion, while our tail-profile bound still yields an explicit, albeit slow, convergence rate.

Our local-SDPI viewpoint is related to, but conceptually distinct from, the classical drift-minorization framework. In that approach, one proves a minorization condition on a small set and a Lyapunov drift condition showing that the chain returns to that set sufficiently often; together, these yield convergence, typically in total variation or weighted total variation \cite{harris1956existence,nummelin1984general,meyn2009markov,rosenthal1995minorization,hairer2011yet}. The drift condition is of the form $\int V(y)\sK(x,dy)\le \lambda V(x)+b\mathbf 1_C(x)$ with $\lambda<1$ for a Lyapunov function $V$ and a small set $C$.
Our approach retains the same core-tail geometry but replaces the drift-minorization certificate by a divergence-contraction certificate: we prove a local SDPI on the core $C_R$ and show that the resulting contraction coefficient directly controls the rejection profile on that set. The final bound separates local contraction from the loss due to localization, expressed through $H_R$, and yields a convergence result in terms of $\mathsf E_\gamma$-divergence for all $\gamma\geq 1$ under warm starts. 

\paragraph{Notation.} Random variables are denoted by uppercase letters, such as $X$. We use calligraphic letters for sets, except for $\mathcal N$, which denotes a Gaussian distribution. For $n\in\mathbb N$, let $[n]:=\{1,\ldots,n\}$. The set of probability measures on a measurable space $\mathcal X$ is denoted by $\mathcal P(\mathcal X)$. A differentiable function $f:\mathbb R^d\to\mathbb R$ is $M$-smooth if $\nabla f$ is $M$-Lipschitz.  A Markov kernel $\mathsf{K}:\mathcal{K}\to\mathcal{P}(\mathcal{W})$ is specified by a collection of distributions $\{\mathsf{K}(x, \cdot)\in\mathcal{P}(\mathcal{W}):x\in\mathcal{K}\}$.  If $\sK:\mathcal X\to\mathcal P(\mathcal Y)$ is a Markov kernel and $\mu\in\mathcal P(\mathcal X)$, then $\mu\sK$ denotes the push-forward measure on $\mathcal Y$, defined by $$\mu\sK(A):=\int_{\mathcal X}\sK(x,A)\mu(dx).$$

\section{Preliminaries} \label{sec:preliminaries} 

\subsection{\texorpdfstring{$\mathsf E_\gamma$}{Egamma}-divergence, SDPI, and mixing time} 
Given a convex function $f$ satisfying $f(1)=0$, and two measures $\mu$ and $\nu$ on a measurable space $\mathcal X$  such that $\mu \ll \nu$, the $f$-divergence between $\mu$ and $\nu$ is defined as:
\begin{align}
    D_f(\mu\|\nu)\coloneqq \int \text{d}\nu f\Big(\frac{\text{d}\mu}{\text{d}\nu}\Big).
\end{align}
Commonly used instances of $f$-divergence include KL divergence $\mathsf{KL}(\mu\|\nu)$, $\chi^2$-divergence $\chi^2(\mu\|\nu)$, total variation distance $\mathsf{TV}(\mu, \nu)$, and Hellinger divergence $\mathcal{H}_\alpha(\mu\|\nu)$ of order $\alpha>1$. These measures are $f$-divergence with associated generator function $f(t)$ to be $t\log t$, $(t-1)^2$, $\frac{1}{2}|t-1|$, and $\frac{t^\alpha - 1}{\alpha - 1}$, respectively. Note that while R\'enyi divergence $D_\alpha$ of order $\alpha$ is not an $f$-divergence, it is a monotone function of $\mathcal{H}_\alpha$, that is,   $D_\alpha(\mu\|\nu)\coloneqq \frac{1}{\alpha - 1} \log\left(1+ (\alpha - 1)\mathcal{H}_\alpha(\mu \parallel \nu)\right)$.

An important instance of $f$-divergence for this work is $\sE_\gamma$-divergence (also known as hockey-stick divergence) defined as $\sE_\gamma(\mu\|\nu)\coloneqq D_{f_\gamma}(\mu\|\nu)$ where $f_\gamma(t) = (t-\gamma)_+$ for $\gamma\geq 1$. It can be verified that 
$$\mathsf E_\gamma(\mu\|\nu)=\sup_{A\subseteq\mathcal X}\{\mu(A)-\gamma\nu(A)\} = \int d(\mu-\gamma\nu)_+,$$
where the supremum is taken over measurable sets $A$ and the equality follows from the Neyman-Pearson lemma. 
Note that at $\gamma=1$, this recovers TV distance: $\mathsf E_1(\mu\|\nu)=\mathsf{TV}(\mu,\nu)$. A key reason to work with $\mathsf E_\gamma$-divergence is that it generates a broad class of $f$-divergences. If $f$ is twice differentiable with continuous second derivative, then \cite[Corollary~3.7]{cohen1998comparisons} \begin{equation}\label{f_representation} D_f(\mu\|\nu)=\int_1^\infty\left[f''(\gamma)\mathsf E_\gamma(\mu\|\nu)+\gamma^{-3}f''(\gamma^{-1})\mathsf E_\gamma(\nu\|\mu)\right]d\gamma. \end{equation} Thus, when both hockey-stick profiles are controlled, one can transfer convergence to the corresponding $f$-divergence whenever the integral is finite. 

A fundamental property of $f$-divergences is the data processing inequality (DPI), which states that each $f$-divergence contracts under Markov kernels: $D_f(\mu\sK\|\nu\sK)\le D_f(\mu\|\nu).$ This inequality can be improved for some kernels $\sK$, that is there may exist $\eta_f \leq 1$ such that $D_f(\mu\sK \| \nu\sK) \leq \eta_f D_f(\mu \| \nu)$ for any measures $\mu$ and $\nu$. The smallest such $\eta_f$ is typically referred to as the \textit{contraction coefficient} of $\sK$ under $f$-divergence and denoted by $\eta_f(\sK)$. If $\eta_f(\sK)<1$, we say $\sK$ satisfies \textit{strong} DPI (SDPI) for $f$-divergence.  In particular, for hockey-stick divergence we write 
\begin{equation}\label{eq:SDPI_gamma}
    \eta_\gamma(\sK):=\sup_{\mu\ne\nu}\frac{\mathsf E_\gamma(\mu\sK\|\nu\sK)}{\mathsf E_\gamma(\mu\|\nu)}.
\end{equation}
We refer to \cite{asoodeh2020contraction,asoodeh2023privacy,Makur_SDPI,raginsky2016strong} for background on contraction coefficients and SDPI. 

Let $\{X_k\}_{k\ge0}$ be a Markov chain with law $\mu_k$ at time $k$ and stationary distribution $\pi$. For $\gamma\ge1$ and $\varepsilon\in(0,1)$, the $\mathsf E_\gamma$-mixing time \cite{Zamanlooy2024} is defined as  $$T_{\mathrm{mix},\mathsf E_\gamma}(\varepsilon):=\min\{k\in\mathbb N:\mathsf E_\gamma(\mu_k\|\pi)\vee \sE_\gamma(\pi\|\mu_k) \le\varepsilon\}.$$ The usual total-variation mixing time is the special case $T_{\mathrm{mix},\mathsf{TV}}(\varepsilon):=T_{\mathrm{mix},\mathsf E_1}(\varepsilon)$.

\SetAlgoNoLine
\LinesNotNumbered

\begin{algorithm}[t]
\caption{Independent Metropolis--Hastings with proposal $q$}
\label{algo:imh}
\KwIn{Initial state $x_0\in\mathcal X$; target density $\pi$; proposal distribution $q$; number of iterations $k$}

\For{$t=0,\dots,k-1$}{
    Draw $y_t\sim q$\;
    Compute $\alpha(x_t,y_t):=1\wedge \frac{w(y_t)}{w(x_t)}$, where $w=d\pi/dq$\;
    Draw $U_t\sim\mathrm{Unif}[0,1]$\;
    Set $x_{t+1}\coloneqq y_t$ if $U_t\le \alpha(x_t,y_t)$, and $x_{t+1}\coloneqq x_t$ otherwise\;
}
\KwOut{Trajectory $\{x_t\}_{t=0}^k$}
\end{algorithm}

\subsection{Projected Langevin Monte Carlo} \label{Sec:PLMC} 
Let $\pi$ be a Gibbs distribution on $\mathbb R^d$ with density proportional to $\exp(-u(x))$, where $u:\mathbb R^d\to\mathbb R$ is a smooth potential. The Langevin diffusion is $$dX_t=-\nabla u(X_t)dt+\sqrt2\,dW_t,$$ where $\{W_t\}_{t\ge0}$ is a standard $d$-dimensional Brownian motion. Under mild regularity conditions on $u$, the distribution of $X_t$ converges to $\pi$ as $t \to \infty$. Applying standard Euler discretization gives Langevin Monte Carlo (LMC): 
\begin{equation}\label{eq:LMC1}
  X_{k+1}=X_k-\eta\nabla u(X_k)+\sqrt{2\eta}\,Z_k,  
\end{equation}
 where $Z_k\sim\mathcal N(0,I_d)$ and $\eta>0$ is the step size. The stationary distribution of the LMC algorithm, denoted by $\pi^\eta$, converges to $\pi$ as $\eta \to 0$; thus, we refer to $\pi^\eta$ as the \textit{biased} target distribution. 
A more general form of the discretized update in \eqref{eq:LMC1} enables handling constrained distributions via projection and large-scale finite-sum potentials $u = \sum_{i=1}^n u_i$ via stochastic gradient and arbitrary batching scheme. 
\begin{definition}[Projected Langevin Monte Carlo] \label{projectedLMC} Let $\mathcal K\subset\mathbb R^d$ be compact and convex, and let potential $u = \sum_{i=1}^nu_i$ with $u_1,\ldots,u_n:\mathcal K\to\mathbb R$ be smooth. Given a random nonempty batch $B_k\subseteq[n]$, define 
\begin{equation}\label{eq:updatePLMC}
    \psi_{B_k}(x):=x-\frac{\eta}{|B_k|}\sum_{i\in B_k}\nabla u_i(x).
\end{equation} Projected Langevin Monte Carlo (P-LMC) is the Markov chain 
\begin{equation}\label{eq:PLMC}
    X_{k+1}=\Pi_{\mathcal K}\left(\psi_{B_k}(X_k)+\sqrt{2\eta}\,Z_k\right),\qquad Z_k\sim\mathcal N(0,I_d),
\end{equation} where $\Pi_{\mathcal K}$ is the Euclidean projection onto $\mathcal K$. We denote by  $\mu_k$ the distribution of $X_k$.
\end{definition}

\subsection{Independent Metropolis--Hastings} \label{sec:MH} Independent Metropolis--Hastings (IMH) is one of the most fundamental Markov chain Monte Carlo algorithms. Given a proposal distribution $q$, the algorithm repeatedly draws a candidate state independently of the current position and then accepts or rejects it using a Metropolis correction; see Algorithm~\ref{algo:imh}. This correction guarantees that the resulting Markov chain has the desired target distribution $\pi$ as its invariant distribution. Assuming $\pi\ll q$, the IMH kernel is given as 
\begin{align*}
    \sK(x, dy) = \alpha(x, y) q(dy) + r(x) \, \delta_x(dy),
\end{align*}
where $\alpha(x, y)$ is the acceptance probability of moving from current state $x$ to $y$ and $r(x)$ is the total rejection probability:
$$\alpha(x,y):=1\wedge\frac{w(y)}{w(x)},\qquad r(x):=1-\int\alpha(x,y)q(dy),$$
and the importance weight is $w(x):=\frac{d\pi}{dq}(x).$ It can be verified that $\sK$ is reversible with invariant distribution $\pi$. A large body of classical work studies the regime in which the importance weight is uniformly bounded, yielding global minorization conditions and uniform ergodicity. More precisely, if $W:=\|w\|_\infty<\infty,$ then $\sK(x,\cdot)\ge W^{-1}\pi(\cdot)$ for every $x$, and consequently, according to the Doeblin minorization condition,   the chain contracts at rate $1-W^{-1}$. This condition is often too strong. In many applications, the proposal is a good approximation to the target on most of the target mass, while the ratio $w=d\pi/dq$ may be very large or unbounded on a small tail region.

In contrast, the focus of this paper is the substantially more challenging and practically relevant setting in which the importance weight may be unbounded. To quantify the severity of the tails of $w$, for $R\ge1$ define the core, its target mass, and its tail profile by $C_R:=\{x:w(x)\le R\}$, $h_R:=\pi(C_R),$ and  $H_R:=1-h_R$, respectively. The function $R\mapsto H_R$ measures the amount of target mass lying in regions where the proposal underestimates the target by more than a factor of $R$. It will play a central role throughout our analysis. One natural way to control the tail profile is via moment: if $M_p:=\mathbb E_q[w^p]<\infty$ for some $p>1$, then $$H_R=\pi(w>R)=\mathbb E_q[w\mathbf 1_{\{w>R\}}]\le M_pR^{-(p-1)}.$$ However, our results are formulated directly in terms of $H_R$ and do not require finite moments. This allows us to accommodate a broad range of unbounded-weight regimes, including polynomial, logarithmic, subexponential, and model-specific tail behaviors.

Our analysis relies on the assumption that we have access to a \textit{warm start}. This is a standard assumption in quantitative mixing-time analysis for MCMC; see, for example, \cite{dwivedi2019log,mangoubi2019nonconvex,chewi2021optimal}. 
\begin{definition}[Warm start] Let $L\ge 1$. We say that an initial distribution $\mu_0$ is $L$-warm with respect to $\pi$ if $\mathsf E_L(\mu_0\|\pi)=0$. When the target distribution is clear from context, we simply say that $\mu_0$ is $L$-warm. 
\end{definition} 
It is important to note that warmness is preserved by any $\pi$-invariant Markov kernel. Indeed, if $\mu_0$ is $L$-warm  and $\pi\sK=\pi$, then for every $k\ge0$, the data processing inequality implies that 
$$\sE_L(\mu_k\|\pi) = \sE_L(\mu_0\sK^k\|\pi\sK^k)\leq \sE_L(\mu_0\|\pi)= 0,$$
implying that $\mu_k$ remains $L$-warm.

\section{Global SDPI for Projected LMC}
\label{sec:plmc-global-sdpi}

This section develops a global SDPI analysis for P-LMC. The key observation is that, on a compact set, the Gaussian smoothing step contracts $\mathsf E_\gamma$-divergence uniformly over all pairs of input distributions. This yields a direct mixing-time bound under $\mathsf E_\gamma$-divergence for every $\gamma\ge 1$, without requiring convexity of the potential. We then lift the same estimates to a broad class of $f$-divergences through the integral representation of $f$-divergences in \eqref{f_representation}. 

Let $\mathcal K\subset\mathbb R^d$ be compact and convex, and write $D:=\mathsf{diam}(\mathcal K)$. Given a batch $B\subseteq[n]$, let $\psi_B:\mathcal K\to \mathbb R^d$ be the update function defined in \eqref{eq:updatePLMC}. 
Throughout this section, the batch law is denoted by $\beta_B:=\mathbb P(B_k=B)$. 
The one-step P-LMC kernel is described in \eqref{eq:PLMC} and can be decomposed into three kernels as
\begin{align}\label{Kernel_representation}
    \mathsf{K} = \Pi_\mathcal{K} \circ \mathsf{K}_G^{\sqrt{2\eta}} \circ \Psi,
\end{align}
where $\Psi\coloneqq \sum_{B\subset [n]}\beta_B\psi_B$ is the random batch-gradient kernel\footnote{By abuse of notation, a deterministic function $\psi_B$ can be viewed as a Markov kernel determined by $w\mapsto \delta_{\psi_B(w)}$. Thus, kernel $\Psi$ must be viewed as $\Psi(x,\cdot):=\sum_{B\subseteq[n]}\beta_B\delta_{\psi_B(x)}(\cdot)$.} and $\mathsf K_G^\sigma(y)=\mathcal N(y,\sigma^2 I_d)$ is the Gaussian smoothing kernel. 
Sampling without replacement with batch size $b$ corresponds to $\beta_B=\binom nb^{-1}$ for $|B|=b$, while Poisson sampling corresponds to $\beta_B=p^{|B|}(1-p)^{n-|B|}$, where $p$ is the probability of including each $i\in [n]$ in the batch. 

The decomposition in \eqref{Kernel_representation} isolates the source of contraction. The deterministic drift $\psi_B$ may expand distances when the potential is non-convex, while the projection step is only non-expansive by data processing. The strict contraction comes from the \textit{constrained} Gaussian kernel: a Gaussian kernel whose input is constrained to a compact set.   
\begin{proposition}[\citealp{asoodeh2020contraction}]
\label{prop_eta_hockey_w_constraned}
Let $\mathcal S\subset\mathbb R^d$ be compact, and let $\mathsf K_G^\sigma$ be the constrained Gaussian kernel $\mathsf K_G^\sigma(y)=\mathcal N(y,\sigma^2I_d)$ for $y\in\mathcal S$. Then
$$\eta_\gamma(\mathsf K_G^\sigma)=\theta_\gamma\Big(\frac{\mathsf{diam}(\mathcal S)}{\sigma}\Big),$$
where $$\theta_\gamma(r)\coloneqq Q\Big(\frac{\log\gamma}{r}-\frac r2\Big)-\gamma Q\Big(\frac{\log\gamma}{r}+\frac r2\Big),$$
and $Q(t):=(2\pi)^{-1/2}\int_t^\infty e^{-u^2/2}\,du$.
\end{proposition}

We now state the main result of this section. It turns the one-step contraction coefficient of the constrained Gaussian kernel from the previous proposition into a multi-step convergence bound for P-LMC. The rate is governed by the diameter of the drifted image $\psi_B(\mathcal K)$ before Gaussian smoothing, and therefore captures how the random-batch sampling scheme affects the geometry seen by the Gaussian step.

\begin{theorem}\label{thm:avg-egamma-general}
Assume that each $u_i$ is $M_i$-smooth on $\mathcal K$, and let $\pi^\eta$ be the invariant distribution of the P-LMC kernel and $\mu_k$ be the distribution of its $k$th iterate. Then,  
for every $\gamma\ge 1$ and every $k\ge 0$, we have 
$$\max\{\mathsf{E}_\gamma(\mu_{k} \| \pi^\eta), \mathsf{E}_\gamma(\pi^\eta\| \mu_{k})\}  \le \rho^k_{\gamma,\beta},$$
where 
$$\rho_{\gamma,\beta} \coloneqq \sum_{B} \beta_B \theta_{\gamma}\Big(\frac{D(\eta M_B + 1)}{\sqrt{2\eta}}\Big),$$
and $M_B\coloneqq |B|^{-1}\sum_{i\in B}M_i$ is the smoothness constant associated with the batch gradient map $\psi_B$.
\end{theorem}
\begin{proof}[Proof sketch]
By stationarity of $\pi^\eta$ and convexity of $(\mu,\nu)\mapsto\mathsf E_\gamma(\mu\|\nu)$ (as for any $f$-divergence), we can write 
$$\mathsf E_\gamma(\mu_{k+1}\|\pi^\eta)\le \sum_{B\subset [n]}\beta_B\,\mathsf E_\gamma\big(\mu_k(\mathsf K_G^{\sqrt{2\eta}}\circ\psi_B)\|\pi^\eta(\mathsf K_G^{\sqrt{2\eta}}\circ\psi_B)\big).$$
For a fixed batch $B$, the image of the drift map is $\mathcal S_B:=\psi_B(\mathcal K)$. Since $\psi_B$ is $(1+\eta M_B)$-Lipschitz, we have   
$\mathsf{diam}(\mathcal S_B)\le D(1+\eta M_B).$
Applying Proposition~\ref{prop_eta_hockey_w_constraned} to the Gaussian kernel restricted to $\mathcal S_B$, and using monotonicity of $r\mapsto\theta_\gamma(r)$, gives
$$\mathsf E_\gamma\big(\mu_k(\mathsf K_G^{\sqrt{2\eta}}\circ\psi_B)\|\pi^\eta(\mathsf K_G^{\sqrt{2\eta}}\circ\psi_B)\big)\le \theta_\gamma\Big(\frac{D(1+\eta M_B)}{\sqrt{2\eta}}\Big)\mathsf E_\gamma(\mu_k\|\pi^\eta).$$
Averaging over $B$ yields $\mathsf E_\gamma(\mu_{k+1}\|\pi^\eta)\le\rho_{\gamma,\beta}\mathsf E_\gamma(\mu_k\|\pi^\eta)$. Iterating and using $\mathsf E_\gamma(\mu_0\|\pi^\eta)\le 1$ proves the claim.
\end{proof}
Theorem~\ref{thm:avg-egamma-general} establishes exponential convergence in $\mathsf E_\gamma$-divergence between the law of P-LMC and its stationary distribution $\pi^\eta$ under smoothness alone, without requiring convexity of the potential. The compact projection set is essential for this global argument: it keeps the pre-noise image $\psi_B(\mathcal K)$ bounded, which turns Gaussian smoothing into a global contractive kernel. The theorem also makes the role of the sampling scheme explicit through an average of batch-specific contraction coefficients, rather than only through a worst-case smoothness bound.

When the potentials are convex, the drift map itself becomes non-expansive for a suitable step size. Indeed, if $g$ is convex and $M$-smooth, then $x\mapsto x-\eta\nabla g(x)$ is non-expansive for $0\le\eta\le 2/M$. Consequently, in the convex case one obtains $\mathsf{diam}(\psi_B(\mathcal K))\le D$, improving the non-convex diameter bound $D(1+\eta M_B)$. This refined diameter estimate leads to a sharper convergence bound; see Appendix~\ref{app:convex} for details.

A useful worst-case consequence follows by replacing all batch smoothness constants by $M:=\max_i M_i$.
\begin{corollary}\label{corollary_Egamma_PLMC}
Assume the conditions of Theorem~\ref{thm:avg-egamma-general} and set $M:=\max_i M_i$. Then, for every $\gamma\ge 1$ and every $k\ge 0$,
$$\max\{\mathsf E_\gamma(\mu_k\|\pi^\eta),\mathsf E_\gamma(\pi^\eta\|\mu_k)\}\le \left[\theta_\gamma\Big(\frac{D(1+\eta M)}{\sqrt{2\eta}}\Big)\right]^k.$$
Consequently, for $0<\varepsilon<1$,
$$T_{\mathrm{mix},\mathsf E_\gamma}(\varepsilon)\le \frac{\log\varepsilon}{\log\theta_\gamma\Big(\frac{D(1+\eta M)}{\sqrt{2\eta}}\Big)}.$$
\end{corollary}


This simplified bound no longer distinguishes among batching schemes. It is therefore best interpreted as a compact global guarantee rather than a sharp description of the effect of stochastic gradients. The bound in Theorem~\ref{thm:avg-egamma-general} is more informative when different batches induce substantially different drift diameters.
\begin{remark}
    The average coefficient in Theorem~\ref{thm:avg-egamma-general} can be substantially sharper than the worst-case bound when the batch smoothness constants are heterogeneous. To illustrate this, consider the following example on a non-convex double-well potential. Let $\mathcal K=[a,b]$, where $a<b$, and define $u(x)\coloneqq c(z^2-1/4)^2$ with $z\coloneqq \frac{x-m}{s}$ and $c>0$, where $m\coloneqq \frac{a+b}{2}$ and $s\coloneqq \frac{b-a}{2}$. Let $u_i(x)=w_i u(x)$, where $w_i\ge0$ and $\sum_{i=1}^n w_i=1$. A direct calculation shows that $u$ is $L$-smooth with $L=44c/(b-a)^2$, hence $u_i$ is $M_i$-smooth with $M_i=w_iL$.
Table~\ref{tab:tv-poisson-swr} compares the average-case SDPI bound from Theorem~\ref{thm:avg-egamma-general} with the worst-case bound from Corollary~\ref{corollary_Egamma_PLMC}. We take $n=12$, $a=-1$, $b=1$, $c=0.1$, $\eta=0.15$, $w_1=0.8$, and $w_i=0.2/(n-1)$ for $i=2,\ldots,n$. For Poisson sampling we take $p=0.2$, and for sampling without replacement (SwR) we take batch size $b=2$.
\end{remark}
\begin{table}[t]
\centering
\caption{Average-case and worst-case SDPI bounds for Poisson sampling and sampling without replacement.}
\label{tab:tv-poisson-swr}
\begin{tabular}{|c|c|c|c|}
\hline
$k$ & Poisson average-case & SwR average-case & worst-case\\
\hline
1  & 0.650 & 0.651 & 0.751 \\
\hline 
5  & 0.116 & 0.117 & 0.239\\
\hline 
10 & 0.013 & 0.013 & 0.057\\
\hline 
15 & 0.001 & 0.001 & 0.013\\
\hline 
20 & $1.835\times 10^{-4}$ & $1.885\times 10^{-4}$ & 0.003 \\
\hline
\end{tabular}
\end{table}

Setting $\gamma=1$ in Corollary~\ref{corollary_Egamma_PLMC} gives the corresponding TV estimate, since $\mathsf E_1=\mathsf{TV}$ and $\theta_1(r)=1-2Q(r/2)$.

\begin{corollary}\label{corollary_TV_PLMC}
Under the assumptions of Corollary~\ref{corollary_Egamma_PLMC}, we have for $0<\varepsilon<1$,
$$T_{\mathrm{mix},\mathsf{TV}}(\varepsilon)\le \frac{\log\varepsilon}{\log\Big(1-2Q\Big(\frac{D(1+\eta M)}{2\sqrt{2\eta}}\Big)\Big)}.$$
\end{corollary}

The quality of this mixing time bound is governed by the ratio between the diameter of the drifted set and the Gaussian smoothing scale. When $\mathsf{diam}(\psi_B(\mathcal K))$ is comparable to or smaller than $\sqrt\eta$, Gaussian smoothing produces a substantial one-step contraction. When the drifted set is much larger than the noise scale, the theorem still gives exponential decay, but the numerical rate becomes weak. This is because the only strict contraction in the argument comes from the Gaussian step, and the Gaussian contraction coefficient necessarily deteriorates as the pre-noise diameter grows.

\subsection{From hockey-stick mixing to $f$-divergence mixing}
\label{subsec:plmc-fdiv}
Theorem~\ref{thm:avg-egamma-general} and Corollary~\ref{corollary_Egamma_PLMC} control the full hockey-stick divergence profiles $\{\mathsf E_\gamma(\mu_k\| \pi^\eta):\gamma\ge 1\}$ and $\{\mathsf E_\gamma(\pi^\eta\|\mu_k):\gamma\ge 1\}$. Combining these results with the integral representation of general $f$-divergence in terms of $\sE_\gamma$-divergence in \eqref{f_representation} enables us to derive convergence bounds for a broad class of $f$-divergences. 
\begin{theorem}\label{Theorem_fdivergence_PLMC}
Let $f:(0,\infty)\to\mathbb R$ be twice differentiable and convex, with continuous second derivative and $f(1)=0$. Set $r\coloneqq D(1+\eta M)/\sqrt{2\eta}$ and $s\coloneqq \exp(r^2/2+r)$. 
Assume that there exist constants $L,N<\infty$ and an integer $K\ge 1$ such that, $t^{-2}f''(t^{-1})\le L,$ and $t^{1-K}f''(t)\le N$ for all $t\ge s$. Then, we have for every $k\ge K$,
$$D_f(\mu_k\|\pi^\eta)\le \left[f'(s)-s^{-1}f'(s^{-1})+f(s^{-1})\right]\theta_1(r)^k+\frac{r\big(L+Ne^{Kr^2}\big)}{k-1}(2\pi)^{-k/2}.$$
\end{theorem}
This theorem shows that the hockey-stick convergence bound is not merely a TV estimate in disguise.
Once both $\mathsf E_\gamma$ profiles are controlled, convergence transfers to many $f$-divergences. Note that the assumptions on $f''$ are mild polynomial growth conditions and hold for standard divergences used in literature. For example, they hold with $N = L = 1$ for KL divergence, $N = L = 2$ for $\chi^2$-divergence, and with $N = L = \alpha$ for the $\mathcal{H}_\alpha$-divergence.
Recall that while R\'enyi divergence is not an $f$-divergence, it is a monotone function of $\mathcal{H}_\alpha$-divergence, thus Theorem~\ref{Theorem_fdivergence_PLMC} can be used to derive convergence bound under R\'enyi divergence too. 

\begin{corollary}
\label{corollary_KL_chi_PLMC}
Let $r=D(1+\eta M)/\sqrt{2\eta}$ and $s=\exp(r^2/2+r)$. Under the assumptions of Theorem~\ref{thm:avg-egamma-general}, the following bounds hold.
\begin{itemize}[leftmargin=*]
\item For every $k\ge 2$, we have 
$$\mathsf{KL}(\mu_k\|\pi^\eta)\le \left(\frac{r^2}{2}+r+1-s^{-1}\right)\theta_1(r)^k+\frac{r(1+e^{r^2})}{k-1}(2\pi)^{-k/2},$$
and $$\chi^2(\mu_k\|\pi^\eta)\le \left(2s-1-s^{-2}\right)\theta_1(r)^k+\frac{2r(1+e^{r^2})}{k-1}(2\pi)^{-k/2}.$$
\item For every $k \ge \lceil \alpha\rceil$, we have
        \begin{align}
            D_{\alpha}(\mu_{k} \| \pi^\eta) \leq&\frac{1}{\alpha - 1} \log \Bigg[ \frac{\alpha r \big(1 + e^{\lceil \alpha-1 \rceil r^2}\big)}{(k-1)(\alpha-1)^{-1}} (2\pi)^{-k/2} + \Big(\alpha s^{\alpha-1}-1 -\frac{\alpha - 1}{s^{\alpha}}\Big) \Big( Q\big(\frac{-r}{2}\big) \Big)^k+ 1\Bigg]\nonumber.
        \end{align}
\end{itemize}
\end{corollary}

We end this section with a technical remark.

\begin{remark} 
The one-step contraction coefficient in Proposition~\ref{prop_eta_hockey_w_constraned} cannot be uniformly improved over the full class of smooth non-convex potentials. More precisely, there exist smooth non-convex potentials for which the true one-step contraction coefficient $\eta_\gamma$ of the resulting P-LMC kernel asymptotically matches the bound given by $\theta_\gamma$; see Appendix~\ref{app:optimalityofTheta} for an explicit construction and its analysis. This does not, however, imply that the resulting multi-step convergence bound in Theorem~\ref{thm:avg-egamma-general} is globally optimal. The theorem applies the contraction coefficient \textit{linearly} and uniformly at each step. A potentially sharper analysis may be possible through a nonlinear SDPI profile that tracks how the contraction depends on the current divergence level, rather than through a single worst-case coefficient.
\end{remark}

\section{Local SDPI for Independent Metropolis--Hastings} \label{sec:mixingtime_MH}

The previous section showed that global contraction gives a compelling convergence theory for P-LMC on compact domains: compactness of $\mathcal K$ and Gaussian smoothing step together provide a nontrivial contraction coefficient for each iterate of P-LMC. Independent Metropolis--Hastings (IMH) behaves differently. Let $q$ be the proposal distribution, $\pi$ be the target distribution, and assume $\pi\ll q$. The IMH transition kernel is $$\sK(x,dy)=\alpha(x,y)q(dy)+r(x)\delta_x(dy),\qquad \alpha(x,y):=1\wedge\frac{w(y)}{w(x)},$$ where  the importance weight $w(x):=\frac{d\pi}{dq}(x)$ and $r(x):=1-\int \alpha(x,y)q(dy)$ is the rejection probability. Thus, the transition kernel places mass $r(x)$ at the current state $x$. If $r(x)$ can be arbitrarily close to one, then $\sK(x,\cdot)$ can be arbitrarily close to $\delta_x$. A non-trivial global contraction coefficient is therefore generally unavailable in this regime.

For IMH, this obstruction is governed by the importance weight. When $w$ is unbounded, there are states with large $w(x)$ from which most proposals are rejected, so $r(x)$ can be arbitrarily close to one. This explains why the global contraction mechanism that worked for P-LMC does not provide a useful theory for unbounded-weight IMH. This phenomenon is not specific to IMH. Appendix~\ref{app:MH_SDPI_relationship} shows that, for a general Metropolis--Hastings kernel, a simple sufficient condition for a nontrivial global contraction coefficient is a uniform lower bound on the acceptance probability. Such a condition prevents the transition kernel from placing arbitrarily large mass on the current state, but it is highly restrictive and fails in the unbounded-weight IMH regimes considered here.

We therefore replace global contraction by local contraction. For $R\ge1$, define the weight-truncated core $$C_R:=\{x:w(x)\le R\}.$$ On this set, the proposal dominates the target up to factor $R$, so the IMH kernel admits a nontrivial local contraction mechanism. The cost of restricting to $C_R$ is measured by the stationary tail profile outside the core. We write $$h_R:=\pi(C_R),\qquad H_R:=1-h_R=\pi(w>R),\qquad \pi_R:=\pi(\cdot\mid C_R),$$ and assume throughout that $h_R>0$.
For $\alpha\ge1$, define the local hockey-stick contraction coefficient by $$\rho_\alpha(R):=\sup_{\substack{\nu\in \mathcal P(C_R):\\ \mathsf E_\alpha(\nu\|\pi_R)>0}}\frac{\mathsf E_\alpha(\nu\sK\|\pi_R\sK)}{\mathsf E_\alpha(\nu\|\pi_R)}.$$ This is a localized and reference-based contraction coefficient. In contrast to the global coefficient $\eta_\alpha(\sK)$ in \eqref{eq:SDPI_gamma}, which allows both input measures to vary, $\rho_\alpha(R)$ fixes the reference measure to $\pi_R$ and varies only the input law $\nu$ supported on $C_R$. Thus, it measures how strongly one step of IMH contracts distributions on the core toward the locally averaged transition $\pi_R\sK$.

We now connect this local coefficient to the specific obstruction that controls IMH convergence. For IMH, slow mixing is driven by repeated rejections: if the chain starts from a state $x$ and repeatedly rejects proposals, it remains at $x$. Thus, the pointwise rejection probability $r(x)$ is the local quantity that must be controlled. The next proposition shows that $\rho_\alpha(R)$ provides exactly such a control on the core $C_R$. This connection follows from the atomic part of the IMH transition. Under a non-atomic proposal, the accepted-proposal part of the IMH kernel places no mass on the singleton $\{x\}$. Hence, the only mass that $\sK(x,\cdot)$ assigns to $\{x\}$ is the rejection mass $r(x)$. In contrast, $\pi_R\sK$ has no atom at $\{x\}$. Testing the hockey-stick divergence on the set $\{x\}$ therefore forces $\rho_\alpha(R)$ to be at least $r(x)$. Consequently, $\rho_\alpha(R)$ is not merely an abstract local SDPI coefficient: it directly upper bounds the rejection probability on $C_R$.

\begin{proposition}
\label{prop:rho-controls-rejection}
Assume that $q$ is non-atomic. Then, for every $R\ge 1$, every $\alpha\ge 1$, and every $x\in C_R$, we have $$r(x)\le \rho_\alpha(R).$$
Consequently,
$$\int r(x)^k\pi(dx)\le h_R\rho_\alpha(R)^k+H_R.$$
\end{proposition}

Proposition~\ref{prop:rho-controls-rejection} turns local contraction into rejection-profile control: once $\rho_\alpha(R)$ is bounded, repeated rejections on the core are controlled, and the only remaining contribution comes from the stationary tail $H_R$. It remains to bound the local coefficient itself which is accomplished by the next proposition. 
\begin{proposition}
\label{prop:localSDPI_IMH}
For every $R\ge 1$ and every $\alpha\ge 1$,
$$\rho_\alpha(R)\le 1-\frac{h_R}{R}.$$
\end{proposition}


Combining this upper bound with Proposition~\ref{prop:rho-controls-rejection} gives
\begin{equation}\label{eq:rejectionSDPI}
    \int r(x)^k\pi(dx)\le h_R\rho_\alpha(R)^k+H_R\le h_R\Big[1-\frac{h_R}{R}\Big]^k+H_R\le e^{-kh_R/R}+H_R.
\end{equation}
The bound separates the two sources of error: repeated rejections on the core decay at rate $h_R/R$, while the mass outside the core contributes the tail term $H_R$. The next theorem converts this rejection-profile estimate into a warm-start bound for $\mathsf E_\gamma(\mu_k\|\pi)$.

\begin{theorem}\label{thm:imh-hs-rejection-profile}
Let $\sK$ be the IMH kernel with target $\pi$, proposal $q$, and importance weight $w=d\pi/dq$. Assume that $q$ is non-atomic and that $\mu_0$ is $L$-warm for some $L>1$ and let $\mu_k:=\mu_0\sK^k$.
Then, for every $\gamma\ge 1$, every $\alpha\ge 1$, and every $R\ge 1$, we have 
\begin{equation}\label{eq:mainthm}
    \mathsf E_\gamma(\mu_k\|\pi)\le a_{L,\gamma}(L+1)\Big[h_R\rho_\alpha(R)^k+H_R\Big],
\end{equation}
where $a_{L,\gamma}\coloneqq \frac{(L-\gamma)_+}{L-1}$. In particular, we have 
$$\mathsf E_\gamma(\mu_k\|\pi)\le a_{L,\gamma}(L+1)\left[e^{-kh_R/R}+H_R\right].$$
\end{theorem}


The proof is given in Appendix~\ref{app:imh_rejection_profile}. The theorem separates the two quantities that govern IMH convergence. The local contraction coefficient $\rho_\alpha(R)$ controls the holding probability on the core $C_R=\{w\le R\}$, while the tail profile $H_R$ measures the amount of target mass outside the region where this control is effective. The explicit IMH estimate $\rho_\alpha(R)\le 1-h_R/R$ then gives the concrete rate. Optimizing over $R$ balances the local contraction term $\exp(-kh_R/R)$ against the tail term $H_R$.


\begin{remark}\label{rem:imh-kl-chi} Theorem~\ref{thm:imh-hs-rejection-profile} gives a one-sided hockey-stick profile bound, controlling $\mathsf E_\gamma(\mu_k\|\pi)$ for all $\gamma\ge1$. Unlike the P-LMC result in Section~\ref{sec:plmc-global-sdpi}, this one-sided control does not by itself imply convergence for arbitrary $f$-divergences through the integral representation in \eqref{f_representation}, since that representation also involves the reverse profile $\mathsf E_\gamma(\pi\|\mu_k)$. Nevertheless, under a warm start, standard divergences such as KL and $\chi^2$-divergence can be controlled directly. Indeed, if $\mu$ is $L$-warm with respect to $\pi$, then $\chi^2(\mu\|\pi)\le 2L\TV(\mu,\pi)$ \cite[Lemma~28]{chewi2021optimal}. Since warmness is preserved by the IMH kernel, Theorem~\ref{thm:imh-hs-rejection-profile} gives, for every $R\ge1$ with $h_R>0$, $$\mathsf{KL}(\mu_k\|\pi)\le \chi^2(\mu_k\|\pi)\le 2L(L+1)\left[e^{-kh_R/R}+H_R\right].$$ Thus the local-SDPI rejection-profile bound also yields quantitative KL and $\chi^2$ convergence under warm starts, even though only the forward hockey-stick profile is controlled. \end{remark}

\subsection{Mixing time under moment constraints}
\label{subsec:imh-moment-tail}

We now compare the tail-profile bound above with the recent finite-time IMH analysis of \cite{deligiannidis2024importance}. Their approach is coupling-based: they construct a common-randomness coupling of two IMH chains using the same proposals and acceptance variables, and use the resulting meeting behavior to control total variation distance. In this coupling, convergence is governed by the event that the chain started from the larger importance weight keeps rejecting, which naturally leads to bounds involving the rejection profile $r(x)^k$ and its stationary average $\int r^k\,d\pi$. 

For standard IMH with unbounded importance weights, \cite[Proposition~4.4]{deligiannidis2024importance} gives a pointwise TV bound of the form
$$\mathsf{TV}(\sK^k(x,\cdot),\pi)\lesssim r(x)^k+M_p k^{-(p-1)},$$
where $M_p:=\mathbb E_q[w^p]$ and the hidden constant depends only on $p$. Integrating this estimate against an $L$-warm initialization yields
$$\mathsf{TV}(\mu_0\sK^k,\pi)\lesssim (L+1)M_p k^{-(p-1)}.$$
The next corollary shows that the same polynomial rate follows from Theorem~\ref{thm:imh-hs-rejection-profile}.
\begin{corollary}\label{cor:imh-moment}
Assume $M_p:=\mathbb E_q[w^p]<\infty$ for some $p>1$. Then, under an $L$-warm start $\mu_0$ and for every $k\ge 1$ and every $\gamma\ge 1$, we have 
$$\mathsf E_\gamma(\mu_k\|\pi)\le C_p\,a_{L,\gamma}(L+1)M_p k^{-(p-1)},$$
where $C_p<\infty$ depends only on $p$.
\end{corollary}

This corollary shows that the moment-based rate follows as a direct consequence of the tail-profile bound. However, the formulation in Theorem~\ref{thm:imh-hs-rejection-profile} is more general than a moment-based statement: it is expressed directly in terms of the tail profile $H_R=\pi(w>R)$. Moment assumptions provide one convenient way to control this profile, but they are not intrinsic to the theorem. When sharper or model-specific estimates of $H_R$ are available, they can be inserted directly into the bound. The next example illustrates this flexibility by giving a quantitative convergence rate in a case where no moment $\mathbb E_q[w^p]$ with $p>1$ is finite.

\begin{example}\label{ex:no-finite-moment-imh}
Let $\mathcal X=[e,\infty)$ and define
$$\pi(dx)=\frac{1}{x(\log x)^2}\,dx,\qquad q(dx)=\frac{1}{c x^2(\log x)^2}\,dx,$$
where $c:=\mathbb E_\pi[1/X]$. Then $q$ is a probability measure and
$w(x)=cx.$
A direct calculation gives $\mathbb E_q[w^p]=\infty$ for all $p>1$.
Thus finite-moment polynomial bounds of \cite{deligiannidis2024importance} do not yield a quantitative rate in this example. In contrast, our tail-profile approach still applies and yields an explicit convergence rate. In fact, the tail profile is explicit: $H_R= 1/\log(R/c)$ for $R\geq 1$. Choosing $R = \frac{k}{2\log\log k}$ in Theorem~\ref{thm:imh-hs-rejection-profile} therefore yields $\mathsf E_\gamma(\mu_k\|\pi)\lesssim 1/\log k$ under warm starts. See Appendix~\ref{app:example} for more details.  
\end{example}

This example illustrates the value of the tail-profile formulation. Moment assumptions imply tail bounds through Markov's inequality, but the converse need not hold. Thus a theorem stated directly in terms of $H_R$ is strictly more flexible than one stated only under finite moment assumptions on $w$.

\section{Discussion and Future Work}
\label{sec:discussion}

This paper develops a contraction-based view of MCMC mixing through hockey-stick divergences. The main message is that the useful notion of contraction depends on the geometry of the kernel. When the transition has a genuinely global smoothing component on a bounded domain, as in P-LMC, a global contraction coefficient gives a direct and transparent proof of exponential convergence. This is the case even for smooth non-convex potentials: compactness controls the diameter of the pre-noise image, Gaussian smoothing supplies the strict contraction, and data processing handles the remaining steps.

For kernels on noncompact spaces, a global contraction coefficient can be too crude or even trivial. Independent Metropolis--Hastings with unbounded importance weights illustrates this obstruction: from high-weight states, the chain may reject with probability arbitrarily close to one, making global contraction trivial. The right object is then local contraction. Our IMH result shows that a local contraction coefficient on the core $C_R=\{w\le R\}$ controls the rejection profile on that core, while the stationary tail profile $H_R=\pi(w>R)$ quantifies the cost of localization. This yields a tail-adaptive convergence bound: it recovers the sharp moment-based rates when finite moments are available, but remains applicable in heavy-tailed regimes where no finite moment of order $p>1$ exists.

A natural next step is to extend this contraction viewpoint beyond linear global or local coefficients. For more complicated Metropolis--Hastings kernels, including random-walk MH and MALA, the one-step contraction may depend on the current divergence level. This suggests studying \textit{nonlinear} SDPI profiles of the form
$$F_{\gamma,\sK}(t):=\sup\{\mathsf E_\gamma(\mu\sK\|\pi):\mathsf E_\gamma(\mu\|\pi)\le t\},\qquad t\in[0,1].$$
A bound on $F_{\gamma,\sK}(t)$ over the range visited by the chain would provide  a nonlinear contraction principle, potentially sharper than the linear SDPI obtained by the contraction coefficient, namely, $F_{\gamma,\sK}(t)<\eta_\gamma(\sK) t$. Such profiles could combine local contraction, acceptance geometry, and Lyapunov-type tail control in a single object. Developing usable bounds on $F_{\gamma,\sK}$ for MALA and other non-independent MH kernels is a promising direction for future work.

\bibliographystyle{alpha}
\bibliography{biblio}

\appendices
\section{Additional related work}
\label{whole_literature}

\paragraph{Langevin Monte Carlo and discretization bias.}
The Langevin diffusion (LD) and its discretizations have been studied extensively across statistical physics, Bayesian statistics, optimization, and machine learning. For the former, \cite{bakry2014analysis} demonstrated that the Log-Sobolev inequality (LSI) and the Poincar\'e inequality (PI) imply exponential convergence in KL and $\chi^2$-divergences, respectively. The convergence of LD under other metrics and assumptions has been explored in works such as \cite{Hosseini2023Towards, chewi2022analysis, vempala2019advances}. Langevin Monte Carlo, also known as the unadjusted Langevin algorithm, approximates the Langevin diffusion by an Euler discretization. A central challenge is to separate the mixing of the discretized chain from the bias between its stationary distribution and the target distribution. This distinction is now standard in non-asymptotic analyses of LMC \cite{dalalyan2017theoretical,cheng2018convergence,vempala2019advances,durmus2019analysis,chewi2022analysis}. Many results establish convergence to the continuous-time target by combining a mixing bound for the discretized chain with a discretization-bias estimate.

\paragraph{Convex LMC.} The convergence of LMC is well understood under convexity assumptions. Under smoothness and strong convexity of the potential, \cite{dalalyan2017theoretical} established non-asymptotic convergence guarantees in total variation, with subsequent improvements and refinements in \cite{durmus2016high,dalalyan2019user}. Related analyses were extended to KL-divergence by \cite{cheng2018convergence}. Later work relaxed strong convexity to convexity under additional regularity assumptions \cite{durmus2019analysis,dalalyan2022bounding}. Several papers also study potentials that are non-convex on a bounded region but become strongly convex outside it, yielding quantitative convergence guarantees under dissipativity or tail-growth conditions \cite{cheng2018sharp,cheng2020stochastic,ma2019sampling,majka2020nonasymptotic,zheng2022constrained}. A different but related viewpoint was developed by \cite{liang2024independent}, who studied the decay of dependence between the initialization and the current output distribution for smooth convex potentials. 

\paragraph{Non-convex LMC.} In the unconstrained non-convex setting, the literature has expanded both the range of metrics and the structural assumptions under which convergence can be proved. Existing results cover convergence in $W_1$ \cite{raginsky2017non}, $W_2$ \cite{chau2021stochastic}, KL-divergence \cite{vempala2019advances}, Fisher information \cite{balasubramanian2022towards}, $\chi^2$-divergence and R\'enyi divergence \cite{erdogdu2022chisquared}, and more general $f$-divergences \cite{mitra2025fast}. These results typically rely on additional conditions such as dissipativity, log-Sobolev inequalities, Poincar\'e inequalities, Lata\l{}a--Oleszkiewicz inequalities, modified log-Sobolev inequalities, weak Poincar\'e inequalities, or variants of weak smoothness \cite{vempala2019advances,erdogdu2021convergence,chewi2022analysis,Hosseini2023Towards,nguyen2023unadjusted}. For example, \cite{mitra2025fast} prove exponential convergence in $f$-divergence under smoothness and an $f$-Sobolev inequality, while \cite{cheng2024fast} obtain conditional convergence under local functional inequalities. These works primarily concern unconstrained LMC, whereas our P-LMC result exploits compactness and Gaussian smoothing to obtain a direct contraction argument for the projected discrete chain. 

\paragraph{Convex P-LMC.} The constrained setting is comparatively less developed. For convex potentials, \cite{bubeck2018sampling} gave an early polynomial-time analysis of projected Langevin algorithms for constrained log-concave sampling under Lipschitzness and smoothness assumptions. The sharp mixing-time behavior of P-LMC was later characterized by \cite{altschuler2022resolving}, who proved tight total-variation mixing bounds to the stationary distribution $\pi^\eta$ of the projected discretized chain under convexity and smoothness. Their analysis uses the notion of \textit{shifted divergences} and is closely connected to a recent model in differential privacy known as privacy amplification by iteration \cite{feldman2018privacy}. Our result is not intended to improve their optimal convex rate. Instead, it identifies a simpler global contraction mechanism: on a compact domain, Gaussian smoothing yields a hockey-stick contraction coefficient, while the remaining parts of the P-LMC update are handled by data processing. This gives a direct exponential convergence proof that continues to hold for smooth non-convex potentials. 

\paragraph{Non-convex P-LMC.} The closest projected non-convex result to ours is \cite{lamperski2021projected}, who analyze a stochastic projected Langevin algorithm in $W_1$ distance. Their framework allows random potentials of the form $\nu(x,Z)$ and assumes smoothness of the mean potential, Lipschitzness of the sample gradients, and uniform sub-Gaussian control of the gradient noise. Under these assumptions, they obtain a bound of the form $$W_1(\mathcal L(X_T),\pi_{\bar\nu})\le c_1(\eta\log T)^{1/4}+c_2e^{-\eta c_3T},$$ for suitable constants $c_1,c_2,c_3$. Their proof compares the discrete projected process with a continuous-time process and controls the resulting discretization error. Our analysis is different in both object and method. We work directly with the discrete projected chain and prove convergence to its stationary distribution $\pi^\eta$, rather than comparing the chain to a continuous-time target. Moreover, our bounds are in the hockey-stick divergence profile and hence imply total-variation convergence and, when both profiles are controlled, convergence in several $f$-divergences. Thus the comparison is not a direct rate comparison, since the metrics and limiting distributions differ.

\paragraph{Differential privacy and SDPI.} The connection between sampling and privacy has become increasingly important. Privacy amplification by iteration studies how randomized iterative algorithms contract divergences relevant to differential privacy \cite{feldman2018privacy,asoodeh2020contraction,asoodeh2023privacy,PrivacyAmplificationMixing}. This perspective was significantly sharpened by \cite{altschuler2022privacy}, who developed a notion of shifted-divergence to greatly improve privacy analyses for noisy iterative algorithms such as SGD under convexity assumptions, and was subsequently used by \cite{altschuler2022resolving} to analyze projected Langevin algorithms for convex potentials. Our work follows the same broad program of importing contraction tools from privacy into sampling, but uses a different contraction principle. Specifically, we rely on hockey-stick contraction estimates for noisy iterative maps developed in \cite{asoodeh2020contraction,asoodeh2023privacy}, which do not require convexity of the underlying update map. This allows us to obtain mixing-time guarantees for P-LMC under smoothness alone, including non-convex potentials. Thus, while both approaches transfer ideas from differential privacy to sampling, the underlying privacy tools apply in different regimes: the shifted-divergence tools used by \cite{altschuler2022resolving} are tailored to convexity, whereas the contraction tools used here apply directly to the non-convex setting.

\paragraph{Metropolis--Hastings and independent proposals.}
The Metropolis algorithm was introduced by \cite{metropolis1953equation} and generalized by \cite{Hastings1970monte}. General convergence theory for MH algorithms has often focused on uniform, geometric, or polynomial ergodicity. For IMH, \cite{MengersonTweedie} characterized uniform ergodicity through a global envelope condition between the proposal and the target. For random-walk Metropolis algorithms, \cite{roberts1996geometric} and later works related convergence to the tail behavior of the target. Gradient-informed MH algorithms such as MALA improve high-dimensional scaling; classical diffusion-limit analyses show improved asymptotic scaling \cite{roberts1998optimal}, and recent non-asymptotic work establishes sharp complexity bounds under strong log-concavity and smoothness assumptions \cite{chewi2021optimal}.

\paragraph{Drift-minorization and Harris theory.}
Classical general-state-space Markov chain theory is built around small sets, regeneration, and Lyapunov drift. Harris recurrence and its quantitative refinements show that a Markov chain converges when it returns sufficiently often to a set on which a minorization condition holds \cite{harris1956existence,nummelin1984general,meyn2009markov}. In MCMC, these ideas lead to explicit convergence bounds through drift-minorization conditions: one proves a minorization on a small set and a Lyapunov drift inequality that drives the chain back toward that set \cite{rosenthal1995minorization,baxendale2005renewal,jones2001honest,hairer2011yet}. Our IMH analysis shares the same broad core-tail intuition, but it uses a different certificate. We do not prove a Lyapunov drift inequality. Instead, for the natural core $C_R=\{w\le R\}$, we prove a local hockey-stick contraction coefficient and show that this coefficient directly controls the rejection profile on the core. The global error is then expressed through the stationary tail profile $H_R=\pi(w>R)$, rather than through a return-time or drift estimate. Thus the contribution is not merely replacing $\TV$ by $\mathsf E_\gamma$: the argument replaces the drift-minorization mechanism by a local contraction-plus-tail-profile principle, yielding divergence-profile bounds under warm starts.

\begin{table}[t]
\centering
    \caption{Overview of papers presenting convergence results for Langevin dynamics and related algorithms.}
    \label{tab:my_label}
    \centering
    \begin{tabulary}{\columnwidth}{|>{\centering\arraybackslash}m{25mm}|>{\centering\arraybackslash}m{12mm}|>{\centering\arraybackslash}m{15mm}|>
    {\centering\arraybackslash}m{35mm}|>{\centering\arraybackslash}m{15mm}|C|}
        \hline
        Reference & Algo. & Convex & Other Assumptions & Metric & Type\\
        \hline
        \cite{bakry2014analysis} & LD & No & PI & $\chi^2$ & to target\\
        \hline
        \cite{bakry2014analysis} & LD & No & LSI & KL & to target\\
        \hline
        \cite{vempala2019advances} & LD & No & LSI & \renyi{} & to target\\
        \hline
        \cite{chewi2022analysis} & LD & No & Lata\l{}a–Oleszkiewicz inequality & \renyi{} & to target\\
        \hline
        \cite{chewi2022analysis} & LD & No & Modified LSI & \renyi{} & to target\\
        \hline
        \cite{Hosseini2023Towards} & LD & No & Weak PI, s-H\"older & \renyi{} & to target\\
        \hline
        \cite{dalalyan2017theoretical} & LMC & Strong & \msmooth{} & TV & to target\\
        \hline
        \cite{dalalyan2019user} & LMC & Strong & \msmooth{} & $W_2$ & to target\\
        \hline
        \cite{durmus2016high} & LMC & Strong & \msmooth{} & $W_2$ & to target\\
        \hline
        \cite{cheng2018convergence} & LMC & Strong & \msmooth{} & KL & to target\\
        \hline
        \cite{cheng2018sharp} & LMC & Strong outside a ball & \msmooth{} & $W_1$ & to target\\
        \hline
        \cite{ma2019sampling} & LMC & Strong outside a ball & \msmooth{} & TV & to target\\
        \hline
        \cite{cheng2020stochastic} & LMC & Strong outside a ball & \msmooth{} & $W_1$ & to biased\\
        \hline
        \cite{durmus2019analysis} & LMC & Yes & \msmooth{} & KL & to target\\
        \hline  
        \cite{dalalyan2022bounding} & LMC & Yes & \msmooth{} & $W_q$ & to target\\
        \hline
        \cite{raginsky2017non} & LMC & No & LSI, \msmooth{}, dissipative & $W_2$ & to target\\
        \hline
        \cite{chau2021stochastic} & LMC & No & \msmooth{}, dissipative & $W_1$ & to target\\
        \hline
        \cite{vempala2019advances} & LMC & No & LSI, \msmooth{} & KL & to target\\
        \hline
        \cite{vempala2019advances} & LMC & No & LSI, \msmooth{} & \renyi{} & to biased\\
        \hline
        \cite{vempala2019advances} & LMC & No & PI, \msmooth{} & \renyi{} & to biased\\
        \hline
        \cite{nguyen2023unadjusted} & LMC & No & LSI, $\alpha$-mix weakly smooth & KL & to target\\
        \hline
        \cite{erdogdu2022chisquared} & LMC & No & LSI, \msmooth{}, dissipative & KL & to target\\
        \hline
        \cite{erdogdu2022chisquared} & LMC & No & LSI, \msmooth{}, dissipative & \renyi{} & to target\\
        \hline
        \cite{erdogdu2021convergence} & LMC & No & Modified LSI, s-H\"older, dissipative & KL & to target\\
        \hline
        \cite{chewi2022analysis} & LMC & No & Lata\l{}a–Oleszkiewicz inequality, s-H\"older & \renyi{} & to target\\
        \hline
        \cite{chewi2022analysis} & LMC & No & Modified LSI, s-H\"older & \renyi{} & to target\\
        \hline
        \cite{Hosseini2023Towards} & LMC & No & Weak PI, s-H\"older & \renyi{} & to target\\
        \hline
        \cite{mitra2025fast} & LMC & No & \msmooth{}, $f$-Sobolev Inequality & $f$-divergence & to biased\\
        \hline
        \cite{balasubramanian2022towards} & Average-LMC & No & \msmooth{} & Fisher information & to target\\
        \hline
        \cite{lamperski2021projected} & P-LMC & No & \msmooth{}, Uniform sub-Gaussian gradients & $W_1$ & to target\\
        \hline
        \cite{bubeck2018sampling} & P-LMC & Yes & \msmooth{}, Lipschitz & TV & to target\\
        \hline
        \cite{altschuler2022resolving} & P-LMC & Yes & \msmooth{} & TV & to biased\\
        \hline
        Ours & P-LMC & No & \msmooth{} & $f$-divergence & to biased\\
        \hline
    \end{tabulary}
\end{table}

\section{Proofs Omitted from Section~\ref{sec:plmc-global-sdpi}}

\subsection{Proof of Theorem~\ref{thm:avg-egamma-general}}
\label{Appendix_Egamma_PLMC}
We begin by stating the following proposition. 

\begin{proposition}
\label{prop_diameter}
Let $\mathcal{K}$ be a compact set with diameter $D$, and define $\mathcal{S}_B \coloneqq \psi_B(\mathcal{K})$ where each potential function $u_i$ is $M_i$-smooth for $i \in B$. Let $M_B \coloneqq \mathbb{E}_{i \in B} [M_i] = \frac{1}{|B|} \sum_{i \in B} M_i$ denote the batch-averaged smoothness constant. Then$$\mathsf{diam}(\mathcal{S}_B) \le D(\eta M_B + 1).$$
\end{proposition}
\begin{proof}[Proof of Proposition \ref{prop_diameter}]

By the definition of the update map $\psi_B(w) = w - \frac{\eta}{|B|} \sum_{i \in B} \nabla u_i(w)$ and the triangle inequality, we have:
\begin{align}
\mathsf{diam}(\mathcal{S}_B) &= \sup_{w_1,w_2 \in \mathcal{K}} ||\psi_B(w_2) -\psi_B(w_1)|| \nonumber\\
&\le \sup_{w_1, w_2 \in \mathcal{K}} ||w_2 - w_1|| + \frac{\eta}{|B|} \sum_{i \in B} \sup_{w_1, w_2 \in \mathcal{K}} ||\nabla u_i(w_1) - \nabla u_i(w_2)|| \nonumber\\
&\le D + \frac{\eta}{|B|} \sum_{i \in B} \sup_{w_1, w_2 \in \mathcal{K}} M_i ||w_2 - w_1|| \nonumber\\
&= D + \eta D \left( \frac{1}{|B|} \sum_{i \in B} M_i \right) = D(\eta M_B + 1), \nonumber
\end{align}
where the second inequality follows from the $M_i$-smoothness of each potential $u_i$, and the final equality follows from the definition of $M_B$.
\end{proof}

Using Proposition~\ref{prop_diameter}, we compute \(\mathsf{E}_\gamma(\mu_{k+1} \| \pi^\eta)\) after \(k+1\) iterations, where the initial inputs are sampled from \(\pi^\eta\) and \(\mu_0\):
\begin{align}
    \mathsf{E}_\gamma(\mu_{k+1} \| \pi^\eta)
    &=\mathsf{E}_\gamma\Big(\mu_k\big(\Pi_\mathcal{K} \circ \mathsf{K}_G^{\sqrt{2\eta}} \circ \Psi_k \big)  \big\| \pi^\eta \big(\Pi_\mathcal{K} \circ \mathsf{K}_G^{\sqrt{2\eta}} \circ \Psi_k \big)\Big) \nonumber\\
    &\le \mathsf{E}_\gamma\Big(\mu_k\big(\mathsf{K}_G^{\sqrt{2\eta}} \circ \Psi_k \big)  \big\| \pi^\eta\big(\mathsf{K}_G^{\sqrt{2\eta}} \circ \Psi_k \big)\Big) \nonumber\\
    &\le \sum_{B \subset \left[n\right]} \mathbb{P}(B_k = B)   \mathsf{E}_\gamma \left(\mu_k(\mathsf{K}_G^{\sqrt{2\eta}} \circ \psi_B ) \big\| \pi^\eta(\mathsf{K}_G^{\sqrt{2\eta}} \circ \psi_B)  \right) \nonumber\\
    &\le \sum_{B \subset [n]} \mathbb{P}(B_k = B)  \theta_{\gamma}\left(\frac{\mathsf{diam}(\mathcal{S}_B)}{\sqrt{2\eta}}\right) \mathsf{E}_\gamma\left(\psi_B(\mu_k)\big\| \psi_B(\pi^\eta)\right) \nonumber\\
    &\le \sum_{B \subset [n]} \mathbb{P}(B_k = B)  \theta_{\gamma}\left(\frac{\mathsf{diam}(\mathcal{S}_B)}{\sqrt{2\eta}}\right) \mathsf{E}_\gamma\left(\mu_k\| \pi^\eta\right) \label{eq_modify_for_convex}\\
    &\le \sum_{B \subset [n]} \mathbb{P}(B_k = B)  \theta_{\gamma}\left(\frac{D(\eta M_B + 1)}{\sqrt{2\eta}}\right)  \mathsf{E}_\gamma\left(\mu_k\| \pi^\eta\right) \nonumber\\
    &= \mathsf{E}_\gamma\left(\mu_k\| \pi^\eta\right) \sum_{B \subset [n]} \mathbb{P}(B_k = B)\theta_{\gamma}\left(\frac{D(\eta M_B + 1)}{\sqrt{2\eta}}\right)\nonumber\\
    &= \rho_{\gamma,B}\mathsf{E}_\gamma \left(\mu_k\|\pi^\eta\right)\nonumber ,
\end{align}

where $\rho_{\gamma,B} := \sum_{B} \mathbb{P}(B_k = B) \theta_{\gamma}\left(\frac{D(\eta M_B + 1)}{\sqrt{2\eta}}\right)$. 
The first step follows directly from the definition of the P-LMC Markov kernel in \eqref{Kernel_representation} and the fact that $\pi^\eta$ is its stationary distribution. Next, we apply the data processing inequality (DPI), followed by an application of the convexity of $(P, Q) \mapsto \mathsf{E}_\gamma(P \| Q)$. The subsequent step leverages Proposition~\ref{prop_eta_hockey_w_constraned}, after which DPI is applied again. Proposition~\ref{prop_diameter} then leads to the next step. Finally, factoring out common terms simplifies the expression, and the last step holds as the summation evaluates to one.

By induction, this yields $\mathsf{E}_\gamma(\mu_{k+1} \| \pi^\eta) \le (\rho_{\gamma,B})^{k+1} \mathsf{E}_\gamma(\mu_0 \| \pi^\eta)$. Finally, since \egamma-divergence is trivially bounded by 1, we obtain the desired result.
The same argument with the two arguments reversed gives the bound for $\mathsf E_\gamma(\pi^\eta\|\mu_{k+1})$.

We now turn to proving Corollary~\ref{corollary_Egamma_PLMC}. Specifically, we aim to determine $k$ such that $\mathsf{E}_\gamma(\mu_{k} \| \pi^\eta) \le \varepsilon$, which holds when
\begin{align}
    \Bigg[\theta_{\gamma}\Big(\frac{D(\eta M + 1)}{\sqrt{2\eta}}\Big)\Bigg]^k \le \varepsilon \nonumber
\end{align}
Taking the natural logarithm of both sides, we have
\begin{align}
    k \ge \frac{\log \varepsilon}{\log\bigg(\theta_{\gamma}\Big(\frac{D(\eta M + 1)}{\sqrt{2\eta}}\Big)\bigg)}\nonumber.
\end{align}
As a result
\begin{align}
    T_{mix,\mathsf{E}_\gamma}(\varepsilon) \le \frac{\log \varepsilon}{\log\bigg(\theta_{\gamma}\Big(\frac{D(\eta M + 1)}{\sqrt{2\eta}}\Big)\bigg)} \nonumber.
\end{align}

\subsection{Proof of Theorem~\ref{Theorem_fdivergence_PLMC}}
\label{Appendix_fdivergence_PLMC}
Recall that 
\begin{equation}\label{def:thetaGamma}
    \theta_\gamma(r)\coloneqq Q\Big(\frac{\log\gamma}{r}-\frac r2\Big)-\gamma Q\Big(\frac{\log\gamma}{r}+\frac r2\Big),
\end{equation}
\begin{lemma}
\label{monotonetheta}
    For any $\gamma \ge 1$, the function $\gamma \mapsto \theta_\gamma(r)$ is monotonically decreasing.
\end{lemma}
\begin{proof}
    We use the Leibniz's rule for differentiation under the integral sign to show that $\gamma \mapsto \theta_\gamma(r)$ has negative derivatives. Let $a(\gamma) \coloneqq \frac{\log\gamma}{r} - \frac{r}{2}$ and $b(\gamma) \coloneqq \frac{\log\gamma}{r} + \frac{r}{2}$. We have:
    \begin{align}
       \sqrt{2\pi}\frac{\partial}{\partial \gamma}\theta_\gamma(r) 
        &= -a'(\gamma)e^{-\frac{a^2(\gamma)}{2}}  -\sqrt{2\pi} + \int_{-\infty}^{b(\gamma)} \!\!\!e^{-\frac{u^2}{2}} \,du +\gamma b'(\gamma) e^{-\frac{b^2(\gamma)}{2}} \nonumber\\
        &= \frac{-1}{\gamma r} e^{-\frac{a^2(\gamma)}{2}} - \sqrt{2\pi} + \int_{-\infty}^{b(\gamma)}\!\!\! e^{-\frac{u^2}{2}} \,du + \frac{1}{r} e^{-\frac{b^2(\gamma)}{2}} \nonumber\\
        &= \underbrace{\frac{1}{r} e^{-\frac{\log^2 \gamma}{r^2} - \frac{r^2}{4}} \left( -1 + \gamma^{-1} \right)}_{T_1} - \underbrace{\left(\sqrt{2\pi} - \int_{-\infty}^{b(\gamma)}\!\!\! e^{-\frac{u^2}{2}} \,du \right)}_{T_2}\nonumber
    \end{align}
   Since $\gamma \geq 1$, the term $T_1$ is non-positive, while $T_2$ is positive because the integral is strictly smaller than $\sqrt{2\pi}$. Thus, \( \theta_\gamma(r) \) has negative derivatives with respect to $\gamma$, completing the proof.
\end{proof}

We set $r = \frac{D(\eta M + 1)}{\sqrt{2\eta}}$ and $s = e^{\frac{{r^2}}{2}+r}$. By substituting our upper bound from Corollary~\ref{corollary_Egamma_PLMC} into \eqref{f_representation}, we obtain:
\begin{align}
    D_{f}(\mu_{k} \| \pi^\eta) &\le \int_1^\infty \Big(f''(\gamma) + \gamma^{-3}f''(\gamma^{-1})\Big)\left[\theta_{\gamma}\left(r\right)\right]^{k} \mathrm{d}\!\gamma\nonumber
\end{align}

The previous integral is split as follows:
\begin{align}
    D_{f}(\mu_{k} \| \pi^\eta) &\le \underbrace{\int_1^{s} \left( f''(\gamma) + \gamma^{-3}f''(\gamma^{-1})\right)\left[\theta_{\gamma}\left(r\right)\right]^{k} \mathrm{d}\!\gamma}_{A} + \underbrace{\int_s^{\infty} \left(f''(\gamma) + \gamma^{-3}f''(\gamma^{-1})\right)\left[\theta_{\gamma}\left(r\right)\right]^{k} \!\!\mathrm{d}\!\gamma}_{B}\nonumber
\end{align}

Regarding term $A$, we first use Lemma~\ref{monotonetheta}. Under the assumption that $f$ is twice continuously differentiable, we derive the following upper bound:
\begin{align}
    A &= \int_1^{s} \left( f''(\gamma) + \gamma^{-3}f''(\gamma^{-1})\right)\left[\theta_{\gamma}\left(r\right)\right]^{k} \mathrm{d}\!\gamma \nonumber\\
    &\le \int_1^{s} \left[ f''(\gamma) + \gamma^{-3}f''(\gamma^{-1})\right]\left[\theta_{1}\left(r\right)\right]^{k}\mathrm{d}\!\gamma \nonumber\\
    &= \left[\theta_{1}\left(r\right)\right]^{k} \int_1^{s} \left[f''(\gamma) + \gamma^{-3}f''(\gamma^{-1})\right] \mathrm{d}\!\gamma \nonumber\\
    &= \left[\theta_{1}\left(r\right)\right]^{k}\left[ \int_1^{s} \!\!\!f''(\gamma)\mathrm{d}\!\gamma + \int_{s^{-1}}^{1} tf''(t)\mathrm{d}t\right] \label{eq:substitution}\\
    &= \left[\theta_{1}\left(r\right)\right]^{k}\left[\int_1^{s} \!\!\!f''(\gamma)\mathrm{d}\!\gamma + tf'(t)\Big|_{\frac{1}{s}}^1 - \int_{\tfrac{1}{s}}^1 \!\!\!f'(t)\mathrm{d}t\right] \label{eq:ibp}\\
    &= \left[\theta_{1}\left(r\right)\right]^{k} \left[f'(s) - s^{-1}f'(s^{-1}) - f(1) + f(s^{-1}) \right] \nonumber\\
    &= \left[\theta_{1}\left(r\right)\right]^{k} \left[f'(s) - s^{-1}f'(s^{-1}) + f(s^{-1}) \right] \label{eq:final_bound}
\end{align}

In the above derivation, \eqref{eq:substitution} follows from the substitution $t = \gamma^{-1}$, and \eqref{eq:ibp} is obtained via integration by parts. The final equality \eqref{eq:final_bound} utilizes the property $f(1) = 0$, which holds for all $f$-divergences.

We now derive an upper bound on $B$. First, we simplify $\theta_\gamma(r)$ by dropping the second term and applying the Gaussian tail bound, $Q(x) \le \frac{p(x)}{x}$ for $x > 0$:
\begin{align*}
\theta_\gamma(r) \le Q\left(\tfrac{\log \gamma}{r} - \tfrac{r}{2}\right) \le \frac{p\left(\frac{\log \gamma}{r} - \frac{r}{2}\right)}{\frac{\log \gamma}{r} - \frac{r}{2}}.
\end{align*}

Moreover, under the assumptions
\begin{align*}
\forall x \ge s:\quad x^{1-K} f''(x) \le N,
\qquad
\forall x \ge s:\quad x^{-2} f''(x^{-1}) \le L,
\end{align*}
we obtain
\begin{align}
B &= \int_s^{\infty} \left(f''(\gamma) + \gamma^{-3} f''(\gamma^{-1})\right)
\left[Q\!\left(\tfrac{\log \gamma}{r} - \tfrac{r}{2}\right)\right]^k \, \mathrm{d}\gamma \nonumber \\
&\le \int_s^{\infty} 
\left(\gamma^{K-1} \gamma^{1-K} f''(\gamma) + \gamma^{-3} f''(\gamma^{-1})\right)
\left[\frac{p\!\left(\frac{\log \gamma}{r} - \frac{r}{2}\right)}{\frac{\log \gamma}{r} - \frac{r}{2}}\right]^k
\mathrm{d}\gamma \nonumber \\
&\le \int_s^{\infty} 
\left(N \gamma^{K-1} + L \gamma^{-1}\right)
\left[\frac{p\!\left(\frac{\log \gamma}{r} - \frac{r}{2}\right)}{\frac{\log \gamma}{r} - \frac{r}{2}}\right]^k
\mathrm{d}\gamma.\nonumber
\end{align}

We next apply two successive changes of variables, first $t = \log \gamma$ and then $x = \frac{t}{r} - \frac{r}{2}$. Substituting and using $p(x) = (2\pi)^{-1/2} e^{-x^2/2}$ gives
\begin{align*}
B \le (2\pi)^\frac{-k}{2} \int_{1}^{\infty} r\left(N \left(e^{rx+\frac{r^2}{2}}\right)^{K} + L\right) \left[\frac{\exp\left(-\frac{x^2}{2}\right)}{x}\right]^{k} \mathrm{d}x.
\end{align*}
 
Regrouping terms yields
\begin{align}
B \le r(2\pi)^\frac{-k}{2} \int_{1}^\infty Ne^{Kr^2}\left[\frac{e^{\frac{-(x-r)^2}{2}}}{x}\right]^{K}\left[\frac{e^{-\frac{x^2}{2}}}{x}\right]^{k-K} + L\left[\frac{e^{-\frac{x^2}{2}}}{x}\right]^{k} \mathrm{d}x \nonumber
\end{align}

Finally, using the elementary bound $e^{-y^2} \le 1$ for all $y \ge 1$, we upper bound the exponential terms and obtain
\begin{align}
B &\le r(2\pi)^\frac{-k}{2} \int_{1}^\infty \big(Ne^{Kr^2} + L\big)x^{-k} \mathrm{d}x \nonumber\\
&= r(2\pi)^\frac{-k}{2} \left(Ne^{Kr^2} + L\right) \left(\frac{1}{k-1}\right)\nonumber.
\end{align}

The final step is to combine the upper bounds for $A$, and $B$. This gives us:
\begin{align}
    D_f(\mu_{k} \| \pi^\eta) \leq &\frac{r\big(L + N e^{K r^2}\big)}{k-1} (2\pi)^{\frac{-k}{2}} + \left[ f'(s) - \frac{f'(s^{-1})}{s} + f(s^{-1}) \right] \!\! \Big[\theta_{1}\left(r\right)\Big]^{k} \nonumber
\end{align}

\subsection{Improvement under convexity}
\label{app:convex}
We start by modifying Proposition~\ref{prop_diameter} for the convex case.

\begin{proposition}
\label{prop_diameter_convex}
Let $\mathcal K\subset\mathbb R^d$ be compact with diameter $D$. Assume that each $u_i$ is convex and $M$-smooth on $\mathcal K$, and define
$$\psi_B(x):=x-\frac{\eta}{|B|}\sum_{i\in B}\nabla u_i(x),\qquad \mathcal S_B:=\psi_B(\mathcal K).$$
If $\eta\le 2/M$, then
$$\mathsf{diam}(\mathcal S_B)\le D.$$
\end{proposition}
\begin{proof}
Let $b:=|B|$ and define the averaged batch potential
$$\bar u_B(x):=\frac1b\sum_{i\in B}u_i(x).$$
Then $\bar u_B$ is convex and $M$-smooth, and
$$\psi_B(x)=x-\eta\nabla\bar u_B(x).$$
We show that $\psi_B$ is non-expansive. Fix $x,y\in\mathcal K$. If the $u_i$ are twice differentiable, then by the fundamental theorem of calculus,
$$\psi_B(x)-\psi_B(y)=\left(I-\eta A_{x,y}\right)(x-y),\qquad A_{x,y}:=\int_0^1 \nabla^2\bar u_B(y+t(x-y))dt.$$
Since $\bar u_B$ is convex and $M$-smooth, $A_{x,y}$ is symmetric positive semidefinite and all its eigenvalues lie in $[0,M]$. Hence every eigenvalue of $I-\eta A_{x,y}$ lies in $[1-\eta M,1]$. If $\eta\le2/M$, then
$$\|I-\eta A_{x,y}\|_{\mathrm{op}}\le1.$$
Therefore
$$\|\psi_B(x)-\psi_B(y)\|\le \|x-y\|.$$
For merely $M$-smooth convex potentials, the same conclusion follows by the standard approximation argument, or equivalently by the standard non-expansiveness of the gradient step $I-\eta\nabla f$ for convex $M$-smooth $f$ and $\eta\le2/M$.
Thus $\psi_B$ is $1$-Lipschitz on $\mathcal K$, and consequently
$$\mathsf{diam}(\mathcal S_B)=\sup_{x,y\in\mathcal K}\|\psi_B(x)-\psi_B(y)\|\le \sup_{x,y\in\mathcal K}\|x-y\|=D.$$
\end{proof}

Having Proposition~\ref{prop_diameter_convex}, we revise the upper bound for TV distance and mixing time for P-LMC. A straightforward manipulation of \eqref{eq_modify_for_convex} leads to following bound for TV distance:
\begin{align}
    \mathsf{TV}(\mu_{k}, \pi^\eta) &\le \Bigg[1 - 2Q\Big(\frac{D}{2\sqrt{2\eta}}\Big)\Bigg]^k. \nonumber
\end{align}
This yields the following upper bound for mixing time:
\begin{align}
        T_{mix,\mathsf{TV}}\Big(\varepsilon\Big) \le \frac{\log \varepsilon}{\log\Big[1 - 2 Q\big(\frac{D}{2\sqrt{2\eta}}\big)\Big]}. \nonumber
\end{align}

\subsection{Average-case vs.\ worst-case convergence bound}
\label{app:averageVSworst}
Consider the compact interval $\mathcal K=[a,b]$ with $a<b$ and potential function $u(x)=c\big(z^2-\tfrac14\big)^2$, where $z\coloneqq\frac{x-m}{s}$, $m\coloneqq \frac{a+b}{2}$,   $s\coloneqq \frac{b-a}{2},$ $c>0$ is a fixed constant. Now, we define $u_i(x) = w_i u(x)$ for $i\in [n]$ where the weights $w_i$ satisfy $w_i\geq 0$ and $\sum_{i=1}^nw_i=1$. This is a rescaled double-well potential on the interval $[a,b]$ (thus non-convex), with wells located at $x = m \pm \frac{s}{2}$. It can be verified that $u$ is $L$-smooth with $L = \tfrac{44c}{(b-a)^2}$, thus each $u_i$ is $M_i$-smooth with $M_i = w_i L$.

In Table~\ref{tab:tv-poisson-swr2}, we demonstrate the bounds obtained from Theorem~\ref{thm:avg-egamma-general} for the Poisson sampling and sampling without replacement. We then compare these bounds with the worst-case bound (Corollary~\ref{corollary_Egamma_PLMC}) and also with the approximate values of $\mathsf{TV}(\mu_k, \pi^\eta)$ for different values of $k$. In this example, we use $n= 12$, $a = -b = 1$, $c=0.1$, $\eta=0.15$, and the weights are as follows: $w_1=0.8$ and $w_i=\frac{0.2}{n-1},$ for $i\in \{2,\dots,n\}$. For the Poisson sampling, we take $p = 0.2$ and for sampling without replacement, we take $b = 2$. In the average-case setting, Poisson sampling consistently outperforms sampling without replacement, albeit slightly.
\begin{table}[t]
\centering
\caption{Empirical total variation distance, worst-case bound, and average-case bound for Poisson sampling and sampling without replacement.}
\label{tab:tv-poisson-swr2}
\begin{tabular}{|c | c | c | c | c | c |}
\hline
 & \multicolumn{2}{c}{Poisson sampling} \vline & \multicolumn{2}{c}{Sampling without replacement} \vline & \\
 \cline{2-5}
$k$ & empirical & average-case & empirical & average-case & Worst-Case\\
\hline
1  & 0.110 & 0.650 & 0.110 & 0.651 & 0.751 \\
5  & $3.015\times 10^{-5}$ & 0.116 & $3.000\times 10^{-5}$ & 0.117 & 0.239\\
10 & $1.046\times 10^{-9}$ & 0.013 & $1.035\times 10^{-9}$ & 0.013 & 0.057\\
15 & $3.603\times 10^{-14}$ & 0.001 & $3.566\times 10^{-14}$ & 0.001 & 0.013\\
20 & $2.956\times 10^{-16}$ & $1.835\times 10^{-4}$ & $2.880\times 10^{-16}$ & $1.885\times 10^{-4}$ & 0.003 \\
\hline
\end{tabular}
\end{table}

This example shows that Theorem~\ref{thm:avg-egamma-general} can be conservative for specific benign non-convex potentials. This, however, does \emph{not} mean that the theorem can be uniformly improved over the general class of smooth non-convex potentials as delineated in the next section.

\subsection{Asymptotic optimality of Theorem~\ref{thm:avg-egamma-general}}\label{app:optimalityofTheta}
Consider $\mathcal K_R=[-R,R]$ and $u_{i,R}(x):=-\delta_Rx^2/2$ for $i\in[n]$, where $\delta_R>0$. Each $u_{i,R}$ is smooth and non-convex, with smoothness constant $M_R=\delta_R$. Moreover, every batch $B$ induces the same update map $\psi_{B,R}(x)=(1+\eta\delta_R)x$. Hence, we have $\psi_{B,R}(\mathcal K_R)=[-(1+\eta\delta_R)R,\,(1+\eta\delta_R)R]$, implying $\mathsf{diam}(\psi_{B,R}(\mathcal K_R))=2R(1+\eta\delta_R)$ \textit{exactly}. 
Thus the contraction coefficient from Theorem~\ref{thm:avg-egamma-general} is
$$\rho_{\gamma,R}:=\theta_\gamma\left(\frac{2R(1+\eta\delta_R)}{\sqrt{2\eta}}\right),$$
independently of the sampling scheme. That this is only an upper bound (as opposed to an identity) is merely due to the projection step (following DPI). 

We now compare this upper bound with the asymptotic true one-step divergence of the projected kernel. Choose $\Delta_R>0$ such that $\Delta_R\to\infty$, $\Delta_R/R\to 0$, and $\eta\delta_R R=o(\Delta_R)$ as $R\to \infty$, and define interior initial states $x_R^\pm:=\pm(R-\Delta_R)$. For every batch $B$, each iteration of P-LMC before projection is kernel
$Q_{B,R}(x,\cdot)=\mathcal N\!\bigl((1+\eta\delta_R)x,\,2\eta\bigr),$ while the projected kernel is $\sK_{B,R}=\Pi_{\mathcal K_R}\circ Q_{B,R}$. Let $\mu_R$ be the mean of the unprojected kernel when the initial point is $x_R^+$, i.e., $\mu_R:=(1+\eta\delta_R)(R-\Delta_R).$ Then the distance from the mean $\mu_R$ to the boundary $R$ is $d_R=\Delta_R-\eta\delta_R(R-\Delta_R).$
By assumptions, we have $d_R = \Delta_R + o(\Delta_R)$, or equivalently, $\frac{d_R}{\Delta_R} \to 1$. Thus, $d_R/\sqrt{2\eta}\to\infty$. 

Therefore, we can see that $Q_{B,R}(x^+_R,\cdot)$ falls outside $[-R, R]$ with probability tending to zero, that is, 
$$\mathbb{P}\Big(\mu_R+\sqrt{2\eta}\,Z\notin[-R,R]\Big)\le
2Q\Big(\frac{d_R}{\sqrt{2\eta}}\Big)=o(1).$$
Now, let $(Y, \Pi_{\mathcal K_R}(Y))$ be a coupling of $(Q_{B, R}(x, \cdot), \sK_{B,R}(x, \cdot))$. Then, we have 
$$\TV \bigl(\sK_{B,R}(x_R^+,\cdot),\,Q_{B,R}(x_R^+,\cdot)\bigr)\leq \Pr(Y\notin \mathcal K_R)= o(1).$$
With the same argument, we can also obtain  $\TV \bigl(\sK_{B,R}(x_R^-,\cdot),\,Q_{B,R}(x_R^-,\cdot)\bigr) = o(1).$
Using an identity proved in Appendix~\ref{App:identity}, we can write 
$$\sE_\gamma\!\bigl(\sK_{B,R}(x_R^-,\cdot)\,\|\,\sK_{B,R}(x_R^+,\cdot)\bigr) =\sE_\gamma\!\bigl(Q_{B,R}(x_R^-,\cdot)\,\|\,Q_{B,R}(x_R^+,\cdot)\bigr)+o(1),$$
implying that 
\begin{align*}
    \sE_\gamma\!\bigl(\sK_{B,R}(x_R^-,\cdot)\,\|\,\sK_{B,R}(x_R^+,\cdot)\bigr) & =\theta_\gamma \Big(\frac{2(1+\eta\delta_R)(R-\Delta_R)}{\sqrt{2\eta}}\Big)+o(1)\\
    & = \rho_{\gamma,R}+o(1),
 \end{align*}
 where the second step follows from the fact that $\tfrac{2(1+\eta\delta_R)(R-\Delta_R)}{\sqrt{2\eta}}$ and $\tfrac{2R(1+\eta\delta_R)}{\sqrt{2\eta}}$ approach $\infty$, i.e., $\theta_\gamma$ evaluated at these two arguments differs by $o(1)$. 
Finally, it follows 
$$\rho_{\gamma,R}+o(1) \le\eta_\gamma(\sK_R)\le\rho_{\gamma,R},$$
implying $\eta_\gamma(\sK_R)=\rho_{\gamma,R}+o(1).$ Consequently, the one-step contraction coefficient in Theorem~\ref{thm:avg-egamma-general} is attained asymptotically by a natural family of smooth non-convex potentials. In particular, there is no uniformly smaller replacement for this one-step coefficient over the full class of smooth non-convex potentials without imposing additional structure.

\section{An \texorpdfstring{$\mathsf E_\gamma$}{Egamma}-divergence Identity}\label{App:identity}
\textbf{Claim.} For any distributions $P,P',Q,Q'$ and any $\gamma\ge1$, we have 
$$\bigl|\sE_\gamma(P\|Q)-\sE_\gamma(P'\|Q')\bigr| \le\TV(P,P')+\gamma\,\TV(Q,Q'). $$

To prove this identity, note that $\sE_\gamma(P\|Q)=\sup_A \bigl[P(A)-\gamma Q(A)\bigr]$. Let
$$f(A):=P(A)-\gamma Q(A), \qquad\text{and}\qquad g(A):=P'(A)-\gamma Q'(A).$$
Then, we have 
$$\sE_\gamma(P\|Q)=\sup_A f(A), \qquad\text{and}\qquad  \sE_\gamma(P'\|Q')=\sup_A g(A).$$
Using the elementary inequality
$$\sup_A f(A)-\sup_A g(A)\le \sup_A\bigl(f(A)-g(A)\bigr),$$
we obtain
$$\sE_\gamma(P\|Q)-\sE_\gamma(P'\|Q')\le\sup_A \Bigl[(P(A)-P'(A))-\gamma(Q(A)-Q'(A))\Bigr].$$
Therefore
\begin{align*}
    \sE_\gamma(P\|Q)-\sE_\gamma(P'\|Q')&\le\sup_A |P(A)-P'(A)|+\gamma \sup_A |Q(A)-Q'(A)|\\
    & = \TV(P,P') + \gamma \TV(Q,Q').
\end{align*}
Interchanging \((P,Q)\) and \((P',Q')\) gives
$$\sE_\gamma(P'\|Q')-\sE_\gamma(P\|Q)\le\TV(P,P')+\gamma\,\TV(Q,Q'),$$
and combining the two inequalities yields
$$\bigl|\sE_\gamma(P\|Q)-\sE_\gamma(P'\|Q')\bigr|
\le\TV(P,P')+\gamma\,\TV(Q,Q').$$

\section{Proofs Omitted from Section~\ref{sec:mixingtime_MH}}

\subsection{Global SDPI for general Metropolis-Hastings algorithms}\label{app:MH_SDPI_relationship}
We provide a simple sufficient condition under which a general Metropolis--Hastings kernel admits a nontrivial global contraction coefficient. Let $\pi$ be a target distribution on $\mathcal X$ and let $Q$ be a proposal kernel. The Metropolis--Hastings kernel associated with $(\pi,Q)$ is $$\sK(x,dy)=\alpha(x,y)Q(x,dy)+r(x)\delta_x(dy),$$ where $\alpha(x,y)$ is the Metropolis--Hastings acceptance probability and $$r(x):=1-\int \alpha(x,z)Q(x,dz)$$ is the rejection, equivalently holding, probability. Assume that, for each $x\in\mathcal X$, the measure $Q(x,\cdot)$ admits a density $q(x,\cdot)$ with respect to a common dominating measure, and that $\pi$ admits a density, also denoted by $\pi$. Then the Metropolis--Hastings acceptance probability is given by $$\alpha(x,y):=1\wedge \frac{\pi(y)q(y,x)}{\pi(x)q(x,y)}.$$ Independent Metropolis--Hastings is the special case $Q(x,dy)=q(dy)$. In this case, if $w=d\pi/dq$, then $$\alpha(x,y)=1\wedge\frac{w(y)}{w(x)}.$$

We now show that a uniform lower bound on the acceptance probability yields a global contraction coefficient. Assume that there exists $a>0$ such that $$\alpha(x,y)\ge a,\qquad x,y\in\mathcal X.$$ Then, for every $x$, we can write  $$\sK(x,\cdot)=aQ(x,\cdot)+(1-a)R(x,\cdot),$$ where $$R(x,A):=\frac{1}{1-a}\left\{\int_A(\alpha(x,z)-a)Q(x,dz)+r(x)\delta_x(A)\right\}.$$ The assumption $\alpha(x, y)\ge a$ ensures that $R(x,\cdot)$ is nonnegative. Moreover, $R$ is a Markov kernel, that is, $R(x, \cdot)$ is a probability measure.  

By \cite[Theorem~2]{asoodeh2020contraction}, for any Markov kernel $\sK$ and any $\gamma\ge1$, the hockey-stick contraction coefficient admits a remarkably simple two-point characterization  $$\eta_\gamma(\sK)=\sup_{x,y\in\mathcal X}\mathsf E_\gamma(\sK(x,\cdot)\|\sK(y,\cdot)).$$ Therefore, it suffices to control the one-step divergence between $\sK(x,\cdot)$ and $\sK(y,\cdot)$ for arbitrary $x,y\in\mathcal X$. Using the mixture decomposition above and the joint convexity of $\mathsf E_\gamma$, we obtain $$\mathsf E_\gamma(\sK(x,\cdot)\|\sK(y,\cdot))\le a\mathsf E_\gamma(Q(x,\cdot)\|Q(y,\cdot))+(1-a)\mathsf E_\gamma(R(x,\cdot)\|R(y,\cdot)).$$ Since $\mathsf E_\gamma(P\|Q)\le1$ for all probability measures $P,Q$ and all $\gamma\ge1$, $$\mathsf E_\gamma(\sK(x,\cdot)\|\sK(y,\cdot))\le a\mathsf E_\gamma(Q(x,\cdot)\|Q(y,\cdot))+1-a.$$ Taking the supremum over $x,y$ and applying the same Dobrushin representation to $Q$ gives $$\eta_\gamma(\sK)\le a\eta_\gamma(Q)+1-a=1-a(1-\eta_\gamma(Q)),\qquad \gamma\ge1.$$

\subsection{Proof of Proposition~\ref{prop:rho-controls-rejection}}\label{app:rejectionUB}
Assume $h_R>0$. Since $q$ is non-atomic and $\pi\ll q$, we have $\pi(\{x\})=0$ for every $x\in\mathcal X$. Hence the conditional distribution $\pi_R=\pi(\cdot\mid C_R)$ is also non-atomic: for every $x$,
$$\pi_R(\{x\})=\frac{\pi(\{x\}\cap C_R)}{h_R}=0.$$
Fix $x\in C_R$. For any $z$, the IMH kernel satisfies
$$\sK(z,\{x\})=\int_{\{x\}}\alpha(z,y)q(dy)+r(z)\delta_z(\{x\})=r(z)\mathbf 1_{\{z=x\}},$$
because $q(\{x\})=0$. In particular, $\sK(x,\{x\})=r(x),$
while
$$\pi_R\sK(\{x\})=\int \sK(z,\{x\})\pi_R(dz)=\int r(z)\mathbf 1_{\{z=x\}}\pi_R(dz)=0,$$
since $\pi_R(\{x\})=0$. Moreover, $\mathsf E_\alpha(\delta_x\|\pi_R)=1$, because $\pi_R(\{x\})=0$. Taking $\nu=\delta_x$ in the definition of $\rho_\alpha(R)$ and testing the set $\{x\}$ in the variational formula for $\mathsf E_\alpha$-divergence give
$$\rho_\alpha(R)\ge \mathsf E_\alpha(\delta_x\sK\|\pi_R\sK)\ge \delta_x\sK(\{x\})-\alpha\pi_R\sK(\{x\})=r(x).$$
Therefore $r(x)\le\rho_\alpha(R)$ for every $x\in C_R$. Since $0\le r\le1$,
$$\int r(x)^k\pi(dx)=\int_{C_R}r(x)^k\pi(dx)+\int_{C_R^c}r(x)^k\pi(dx)\le h_R\rho_\alpha(R)^k+H_R.$$

\subsection{Proof of Proposition~\ref{prop:localSDPI_IMH}}\label{app:local_imh_sdpi}
For $x\in C_R$ and $y\in C_R$,
$$\alpha(x,y)q(dy)=\left(1\wedge \frac{w(y)}{w(x)}\right)\frac{\pi(dy)}{w(y)}\ge \frac1R\pi(dy).$$
Indeed, if $w(y)\le w(x)$, then the left-hand side is $\pi(dy)/w(x)\ge \pi(dy)/R$; if $w(y)>w(x)$, then it is $\pi(dy)/w(y)\ge \pi(dy)/R$. Hence, for every $x\in C_R$,
$$\sK(x,\cdot)\ge \frac1R\pi(\cdot\cap C_R)=\frac{h_R}{R}\pi_R(\cdot).$$
Set $\beta_R:=h_R/R$. Then, for $x\in C_R$,
$$\sK(x,\cdot)=\beta_R\pi_R(\cdot)+(1-\beta_R)\widetilde \sK_R(x,\cdot),$$
for some Markov kernel $\widetilde \sK_R$. Therefore, for any $\nu(C_R)=1$,
$$\nu \sK=\beta_R\pi_R+(1-\beta_R)\nu\widetilde \sK_R,\qquad \pi_R\sK=\beta_R\pi_R+(1-\beta_R)\pi_R\widetilde \sK_R.$$
Using the variational formula
$$\mathsf E_\alpha(P\|Q)=\sup_{0\le f\le 1}\{P(f)-\alpha Q(f)\},$$
we get, for every $0\le f\le 1$,
$$\nu \sK(f)-\alpha\pi_R\sK(f)=\beta_R(1-\alpha)\pi_R(f)+(1-\beta_R)\{\nu\widetilde \sK_R(f)-\alpha\pi_R\widetilde \sK_R(f)\}.$$
Since $\alpha\ge 1$, the first term is nonpositive. Thus
$$\mathsf E_\alpha(\nu \sK\|\pi_R\sK)\le (1-\beta_R)\mathsf E_\alpha(\nu\widetilde \sK_R\|\pi_R\widetilde \sK_R)\le (1-\beta_R)\mathsf E_\alpha(\nu\|\pi_R),$$
where the last step is the data processing inequality. Therefore $\rho_\alpha(R)\le 1-h_R/R$.

\subsection{A direct local-to-global SDPI recursion}
\label{app:local_to_global_recursion}
The main text uses the sharper rejection-profile route for IMH. For completeness, we record here a direct local-to-global SDPI recursion. This result is useful conceptually, but for IMH it pays the tail cost at every step and is therefore weaker than Theorem~\ref{thm:imh-hs-rejection-profile}.

The next result lifts the local contraction coefficient $\rho_\alpha$ to the global hockey-stick divergence and shows how the local contraction coefficient on $C_R$ controls the global hockey-stick divergence after one MH step, up to two explicit costs: the mass of the initialization outside the core and the mismatch between the core masses of $\mu$ and $\pi$.

\begin{theorem}[Local-to-global hockey-stick lifting]
\label{thm:Local_to_global}
Let $\mu$ be a probability measure and let $\mu_R:=\mu(\cdot\mid C_R)$ whenever $m_R:=\mu(C_R)>0$. Then, for every $\gamma\ge 1$,
$$\mathsf E_\gamma(\mu \sK\|\pi)\le m_R\rho_{\bar\alpha_R}(R)\mathsf E_{\bar\alpha_R}(\mu_R\|\pi_R)+\mu(C_R^c)+(m_R-\gamma h_R)_+,$$
where
$\bar\alpha_R:=\max\Big\{1,\frac{\gamma h_R}{m_R}\Big\}.$
\end{theorem}
\begin{proof}
By stationarity of $\pi$,
$$\pi=\pi \sK=h_R\pi_R\sK+H_R\pi_{R^c}\sK,$$
where $\pi_{R^c}:=\pi(\cdot\mid C_R^c)$ when $H_R>0$. Similarly,
$$\mu \sK=m_R\mu_R\sK+(1-m_R)\mu_{R^c}\sK.$$
For $0\le f\le 1$,
$$\mu \sK(f)-\gamma\pi(f)=m_R\mu_R\sK(f)-\gamma h_R\pi_R\sK(f)+(1-m_R)\mu_{R^c}\sK(f)-\gamma H_R\pi_{R^c}\sK(f).$$
The last two terms are bounded above by $1-m_R=\mu(C_R^c)$. For the first two terms,
$$m_R\mu_R\sK(f)-\gamma h_R\pi_R\sK(f)=m_R\{\mu_R\sK(f)-\bar\alpha_R\pi_R\sK(f)\}+(m_R\bar\alpha_R-\gamma h_R)\pi_R\sK(f).$$
Since $0\le\pi_R\sK(f)\le1$ and $m_R\bar\alpha_R-\gamma h_R=(m_R-\gamma h_R)_+$,
$$\mu \sK(f)-\gamma\pi(f)\le m_R\mathsf E_{\bar\alpha_R}(\mu_R\sK\|\pi_R\sK)+\mu(C_R^c)+(m_R-\gamma h_R)_+.$$
Taking the supremum over $0\le f\le 1$ and applying the definition of $\rho_{\bar\alpha_R}(R)$ proves the claim.
\end{proof}
The level shift $\bar\alpha_R$ is the price of conditioning on the core. If $\mu$ places less mass on $C_R$ than $\pi$ does, then the density ratio inside $C_R$ is amplified, and the relevant local hockey-stick level can be larger than the global level $\gamma$.

Theorem~\ref{thm:Local_to_global} explains why a direct step-by-step local-SDPI recursion is valid but not sharp. Indeed, one has the deterministic comparison
$$m_R\mathsf E_{\bar\alpha_R}(\mu_R\|\pi_R)\le \mathsf E_\gamma(\mu\|\pi).$$
To see this, let $s=d\mu/d\pi$. On $C_R$,
$$\frac{d\mu_R}{d\pi_R}=\frac{h_R}{m_R}s,$$
and hence
$$m_R\mathsf E_{\bar\alpha_R}(\mu_R\|\pi_R)=\int_{C_R}\left(s-\frac{m_R\bar\alpha_R}{h_R}\right)_+d\pi\le \int_{C_R}(s-\gamma)_+d\pi\le \mathsf E_\gamma(\mu\|\pi).$$
Combining this comparison with Proposition~\ref{prop:localSDPI_IMH} gives the affine one-step bound
$$\mathsf E_\gamma(\mu \sK\|\pi)\le \left(1-\frac{h_R}{R}\right)\mathsf E_\gamma(\mu\|\pi)+\mu(C_R^c)+(m_R-\gamma h_R)_+.$$
If $\mu_0$ is $L$-warm, then so is $\mu_k$ for every $k$, and direct iteration yields
$$\mathsf E_\gamma(\mu_k\|\pi)\le e^{-kh_R/R}\mathsf E_\gamma(\mu_0\|\pi)+\frac{R}{h_R}\{LH_R+[\gamma H_R-(\gamma-1)]_+\}.$$
This is the purely local-SDPI recursion. It is useful as a diagnostic bound, but it pays the tail cost at every step and therefore produces an accumulated tail term of order $RH_R$.

For IMH, one can do better by exploiting the accepted/holding structure of the chain. The next lemma controls the endpoint TV distance through the rejection profile. It charges the tail through the probability of repeated rejections, rather than through a fresh core-tail split at each step.

\subsection{Proof of Theorem~\ref{thm:imh-hs-rejection-profile}}
\label{app:imh_rejection_profile}

We first record the rejection-profile estimate used in the proof of Theorem~\ref{thm:imh-hs-rejection-profile}.
\begin{lemma}[Rejection-profile estimate] \label{lem:rejection-profile} For every initial distribution $\mu_0$ and every $k\ge0$, $$\mathsf{TV}(\mu_0\sK^k,\pi)\le \int r(x)^k\,\mu_0(dx)+\int r(x)^k\,\pi(dx).$$ Consequently, if $\mu_0$ is $L$-warm with respect to $\pi$, then $$\mathsf{TV}(\mu_0\sK^k,\pi)\le (L+1)\int r(x)^k\,\pi(dx).$$ \end{lemma}
\begin{proof}[Proof of Lemma~\ref{lem:rejection-profile}]
    We first prove the pointwise estimate $$\mathsf{TV}(\sK^k(x,\cdot),\sK^k(y,\cdot))\le \max\{r(x),r(y)\}^k,\qquad x,y\in\mathcal X.$$ Fix $x,y\in\mathcal X$ and assume without loss of generality that $w(x)\le w(y)$. We construct two IMH chains $\{X_t\}_{t\ge0}$ and $\{Y_t\}_{t\ge0}$ with $X_0=x$ and $Y_0=y$, using the same proposals and the same acceptance variables. At each step $t$, draw $Z_t\sim q$ and $U_t\sim\mathrm{Unif}[0,1]$, independently over time. The chains update by accepting $Z_t$ whenever $$U_t\le \alpha(X_t,Z_t),\qquad U_t\le \alpha(Y_t,Z_t),$$ respectively, where $$\alpha(u,z):=1\wedge\frac{w(z)}{w(u)}.$$
The key monotonicity property is that, for every proposal $z$, $$w(x)\le w(y)\quad\Longrightarrow\quad \alpha(x,z)\ge \alpha(y,z).$$ Indeed, $$\frac{w(z)}{w(x)}\ge \frac{w(z)}{w(y)},$$ and the map $t\mapsto 1\wedge t$ is nondecreasing. Therefore, at time $0$, whenever the chain started from $y$ accepts the proposal $Z_0$, the chain started from $x$ also accepts the same proposal. In that case both chains move to $Z_0$ and meet. 
More generally, before the two chains meet, the chain started from $y$ can only remain at $y$. If it accepts at some step, then the other chain also accepts the same proposal and the chains meet. Hence the event that the chains have not met by time $k$ is contained in the event that the chain started from $y$ rejects the first $k$ proposals. Since, on this event, that chain remains at $y$ throughout, each rejection has probability $r(y)$, independently from step to step. Thus $$\mathbb P(X_k\ne Y_k)\le r(y)^k.$$ 
Therefore, we have  $$\mathsf{TV}(\sK^k(x,\cdot),\sK^k(y,\cdot))\le \mathbb P(X_k\ne Y_k)\le r(y)^k.$$ Since $w(x)\le w(y)$ implies $r(x)\le r(y)$, we have $r(y)=\max\{r(x),r(y)\}$, and therefore $$\mathsf{TV}(\sK^k(x,\cdot),\sK^k(y,\cdot))\le \max\{r(x),r(y)\}^k.$$

We now pass from the pointwise estimate to convergence to stationarity. Since $\pi$ is invariant for $\sK$, we have $$\pi=\pi\sK^k=\int \sK^k(y,\cdot)\pi(dy).$$ By convexity of total variation in its second argument, $$\mathsf{TV}(\sK^k(x,\cdot),\pi)\le \int \mathsf{TV}(\sK^k(x,\cdot),\sK^k(y,\cdot))\pi(dy).$$ Using the pointwise estimate, $$\mathsf{TV}(\sK^k(x,\cdot),\pi)\le \int \max\{r(x),r(y)\}^k\pi(dy).$$ Since $0\le r\le1$, we have $\max\{a,b\}^k\le a^k+b^k$ for all $a,b\in[0,1]$. Hence $$\mathsf{TV}(\sK^k(x,\cdot),\pi)\le r(x)^k+\int r(y)^k\pi(dy).$$ Integrating this inequality with respect to $\mu_0(dx)$ gives $$\mathsf{TV}(\mu_0\sK^k,\pi)\le \int r(x)^k\,\mu_0(dx)+\int r(y)^k\,\pi(dy).$$ Finally, if $\mu_0$ is $L$-warm, then $\mu_0\le L\pi$, and therefore $$\int r(x)^k\,\mu_0(dx)\le L\int r(x)^k\,\pi(dx).$$ Substituting this into the previous display yields $$\mathsf{TV}(\mu_0\sK^k,\pi)\le (L+1)\int r(x)^k\,\pi(dx).$$ This proves the lemma. 
\end{proof}

The proof of this lemma uses only the monotonicity of the IMH acceptance probability: for a fixed proposal $z$, the map $u\mapsto 1\wedge w(z)/u$ is nonincreasing in the current weight $u$. Thus, under common proposals and common acceptance variables, the lower-weight chain accepts whenever the higher-weight chain accepts.

We now prove Theorem~\ref{thm:imh-hs-rejection-profile}.

\begin{proof}[Proof of Theorem~\ref{thm:imh-hs-rejection-profile}]
By Lemma~\ref{lem:rejection-profile} and the $L$-warmness assumption,
$$\mathsf E_1(\mu_k\|\pi)=\mathsf{TV}(\mu_k,\pi)\le (L+1)\int r(x)^k\pi(dx).$$
Since $\mu_0\le L\pi$ and $\sK$ is $\pi$-invariant, warmness is preserved:
$$\mu_k=\mu_0\sK^k\le L\pi\sK^k=L\pi.$$
Let $s_k:=d\mu_k/d\pi$. Then $0\le s_k\le L$. For $1\le\gamma<L$ and $L>1$, the pointwise inequality
$$ (s_k-\gamma)_+\le \frac{L-\gamma}{L-1}(s_k-1)_+ $$
gives
$$\mathsf E_\gamma(\mu_k\|\pi)=\int(s_k-\gamma)_+d\pi\le \frac{L-\gamma}{L-1}\int(s_k-1)_+d\pi=\frac{L-\gamma}{L-1}\mathsf E_1(\mu_k\|\pi).$$
If $\gamma\ge L$, then $(s_k-\gamma)_+=0$ $\pi$-a.s. If $L=1$, then $\mu_0=\pi$ and the claim is trivial. Therefore, for all $\gamma\ge1$,
$$\mathsf E_\gamma(\mu_k\|\pi)\le a_{L,\gamma}\mathsf E_1(\mu_k\|\pi)\le a_{L,\gamma}(L+1)\int r(x)^k\pi(dx).$$
By Proposition~\ref{prop:rho-controls-rejection},
$$\int r(x)^k\pi(dx)\le h_R\rho_\alpha(R)^k+H_R.$$
This proves
$$\mathsf E_\gamma(\mu_k\|\pi)\le a_{L,\gamma}(L+1)\{h_R\rho_\alpha(R)^k+H_R\}.$$
Finally, Proposition~\ref{prop:localSDPI_IMH} gives $\rho_\alpha(R)\le 1-h_R/R$, and hence
$$h_R\rho_\alpha(R)^k+H_R\le h_R\left(1-\frac{h_R}{R}\right)^k+H_R\le \exp\left(-\frac{kh_R}{R}\right)+H_R.$$
This proves the explicit bound.
\end{proof}

\subsection{Proof of Corollary~\ref{cor:imh-moment}} \label{app:proofCor:imh-moment}
For every $R\ge 1$,
$$H_R:=\pi(w>R)=\mathbb E_q[w\mathbf 1_{\{w>R\}}]\le M_pR^{-(p-1)}.$$
Since $\mathbb E_q[w]=1$, Jensen's inequality gives $M_p\ge1$. Define $R_0:=(2M_p)^{1/(p-1)}.$
Then $R_0\ge1$ and $H_{R_0}\le1/2$.

We first derive a pointwise rejection bound. If $w(x)\le R_0$, then $x\in C_{R_0}$ and $h_{R_0}\ge1/2$, so Proposition~\ref{prop:rho-controls-rejection} and Proposition~\ref{prop:localSDPI_IMH} imply
$$r(x)\le \rho_\alpha(R_0)\le 1-\frac{1}{2R_0}.$$
Hence, if $w(x)\le R_0$, then we have 
$$r(x)^k\le \exp\Big(-\frac{k}{2R_0}\Big).$$
If $w(x)>R_0$, set $R_x:=w(x)$. Then $H_{R_x}\le M_pR_x^{-(p-1)}\le M_pR_0^{-(p-1)}=1/2$, so $h_{R_x}\ge1/2$. Since $x\in C_{R_x}$, Propositions~\ref{prop:rho-controls-rejection} and~\ref{prop:localSDPI_IMH} give
$$r(x)\le \rho_\alpha(R_x)\le 1-\frac{1}{2w(x)},$$
and therefore, if $w(x)>R_0$, then we have 
$$r(x)^k\le \exp\Big(-\frac{k}{2w(x)}\Big).$$
Combining the two cases, we obtain 
$$r(x)^k\le \exp\Big(-\frac{k}{2R_0}\Big)\mathbf 1_{\{w(x)\le R_0\}}+\exp\Big(-\frac{k}{2w(x)}\Big)\mathbf 1_{\{w(x)>R_0\}}.$$
Integrating with respect to $\pi$ gives
$$\int r(x)^k\pi(dx)\le \exp\Big(-\frac{k}{2R_0}\Big)+\mathbb E_\pi\Big[\exp\Big(-\frac{k}{2w}\Big)\mathbf 1_{\{w>R_0\}}\Big].$$
Let $\varphi(t):=\exp(-k/(2t))$. Since $\varphi$ is increasing and $\varphi'(t)=\frac{k}{2t^2}\exp(-k/(2t))$, the tail integration formula gives
$$\mathbb E_\pi[\varphi(w)\mathbf 1_{\{w>R_0\}}]\le \varphi(R_0)H_{R_0}+\int_{R_0}^\infty H_t\varphi'(t)dt.$$
Using $H_t\le M_pt^{-(p-1)}$ and the change of variables $u=k/(2t)$,
$$\int_{R_0}^\infty H_t\varphi'(t)dt\le M_p\int_{R_0}^\infty t^{-(p-1)}\frac{k}{2t^2}e^{-k/(2t)}dt\le M_p\Big(\frac2k\Big)^{p-1}\Gamma(p),$$
where $\Gamma(p)$ is the Gamma function. 
Thus
$$\int r(x)^k\pi(dx)\le \frac32\exp\Big(-\frac{k}{2R_0}\Big)+M_p2^{p-1}\Gamma(p)k^{-(p-1)}.$$
Finally,
$$\exp\Big(-\frac{k}{2R_0}\Big)\le 2^p(p-1)^{p-1}e^{-(p-1)}M_p k^{-(p-1)},\qquad k\ge1,$$
because $R_0^{p-1}=2M_p$. Hence there exists a constant $C_p<\infty$, depending only on $p$, such that
$$\int r(x)^k\pi(dx)\le C_pM_p k^{-(p-1)}.$$
The result follows from Theorem~\ref{thm:imh-hs-rejection-profile}.
\subsection{More details on Example~\ref{ex:no-finite-moment-imh}: Mixing time under infinite moments}\label{app:example}
Let $\mathcal X=[e,\infty)$ and define
$$\pi(dx)=\frac{1}{x(\log x)^2}\,dx,\qquad q(dx)=\frac{1}{c x^2(\log x)^2}\,dx,$$
where $c:=\mathbb E_\pi[1/X]$. Then $q$ is a probability measure and
$w(x)\coloneqq \frac{d\pi}{dq}(x)=cx.$
Note that for every $p>1$:
$$\mathbb E_q[w^p]=\int (cx)^p q(dx)=c^{p-1}\int x^{p-1}\pi(dx)=c^{p-1}\int_e^\infty\frac{x^{p-2}}{(\log x)^2}\,dx=\infty.$$

Thus no moment assumption $\mathbb E_q[w^p]<\infty$ with $p>1$ holds, implying the framework of \cite{deligiannidis2024importance} is not applicable. Nevertheless, the tail profile is explicit.  For $R\ge 1$, we have 
$$H_R:=\pi(w>R)=\pi(X>R/c)=\int_{R/c}^\infty\frac{1}{x(\log x)^2}\,dx=\frac{1}{\log(R/c)}.$$
Theorem~\ref{thm:imh-hs-rejection-profile} therefore gives, for every $L$-warm start and every $\gamma\ge1$,
$$\mathsf E_\gamma(\mu_k\|\pi)\le a_{L,\gamma}(L+1)\Big[e^{-kh_R/R}+H_R\Big].$$
Now choose $R_k:=\frac{k}{2\log\log k}.$ For all sufficiently large $k$, $R_k\ge 1$ and $H_{R_k}\le1/2$, hence $h_{R_k}\ge1/2$. Therefore,
$$e^{-\frac{kh_{R_k}}{R_k}}\le e^{-\frac{k}{2R_k}}=\frac1{\log k},$$
and
$$H_{R_k}=\frac{1}{\log(R_k/c)}=\frac{1}{\log k-\log(2c\log\log k)}\le \frac2{\log k}$$
for all sufficiently large $k$. Substituting $R=R_k$ gives
$$\mathsf E_\gamma(\mu_k\|\pi)\le \frac{3a_{L,\gamma}(L+1)}{\log k}$$
for all sufficiently large $k$. Hence $\mathsf E_\gamma(\mu_k\|\pi)\lesssim 1/\log k.$

\end{document}